\documentclass[conference]{IEEEtran}
\IEEEoverridecommandlockouts
\usepackage[utf8]{inputenc}\usepackage{amsmath,amssymb,amsfonts}
\usepackage{graphicx}
\usepackage{textcomp}
\usepackage{comment}
\usepackage{multicol}
\usepackage{xcolor}
\usepackage{paralist} 
\usepackage[lofdepth,lotdepth]{subfig}
\usepackage{cite}
\usepackage{amsmath,amssymb,amsfonts}
\usepackage{graphicx}
\usepackage{textcomp}
\usepackage{xcolor}
\def\BibTeX{{\rm B\kern-.05em{\sc i\kern-.025em b}\kern-.08em
    T\kern-.1667em\lower.7ex\hbox{E}\kern-.125emX}}
\usepackage{epstopdf}    
\usepackage{nomencl}
\makenomenclature
\usepackage{lscape}

\usepackage{booktabs}
\usepackage{resizegather}

\usepackage{todonotes}
\usepackage{amsmath}
\usepackage{hyperref}
\usepackage{adjustbox}

\usepackage{amssymb}
\usepackage{acronym}
\usepackage{amsthm}
\usepackage{mathtools} 

\usepackage{tikz}
\usetikzlibrary{arrows,automata}
\usetikzlibrary{shapes.geometric,fit}
\tikzstyle{container} = [draw, rectangle, inner sep=1.5cm]

  \tikzset{main node/.style={circle,draw,minimum size=1cm,inner sep=0pt},}
\usepackage{algorithm}
\usepackage{algpseudocode} 
\usepackage{listings}

\newif\ifuseboldmathops
\newif\ifuseittextabbrevs
\useboldmathopstrue   % comment out to use mathbb
\ifuseittextabbrevs

	\newcommand{\ie}{{\it i.e.}}

\else

	\newcommand{\ie}{i.e.}

\fi

\ifuseboldmathops
\else

\fi

\ifuseboldmathops

\else

\fi

\ifuseboldmathops

\else

\fi

\ifuseboldmathops

\else

\fi

\newcommand{\truev}{\mathsf{true}}

\newcommand{\Always}{\Box \, }
\newcommand{\Eventually}{\Diamond \, }

\newcommand{\until}{\mbox{$\, {\sf U}\,$}}

\newcommand{\abs}[1]{\lvert#1\rvert}

\newcommand{\dist}[1]{\mathcal{D}(#1)}

\newcommand{\supp}{\mbox{Supp}}

\newcommand{\calF}{\mathcal{F}}
\newcommand{\calAP}{\mathcal{AP}}

\newcommand{\last}{\mathsf{Last}}

\newcommand{\win}{\mathsf{Win}} 

\acrodef{mdp}[MDP]{Markov Decision Process}
\acrodef{pomdp}[POMDP]{Partially Observable Markov Decision Process}

\theoremstyle{definition}
 \newtheorem{definition}{Definition}
 \newtheorem{example}{Example}
\newtheorem{problem}{Problem}
\newtheorem{remark}{Remark}

\newtheorem{lemma}{Lemma}

 \newtheorem{theorem}{Theorem}

\newcommand{\hgame}{\mathcal{HG}}

\newcommand{\pre}{\mathsf{Pre}}

\newcommand{\calA}{\mathcal{A}}
\newcommand{\game}{\mathcal{G}}

\acrodef{dfa}[DFA]{Deterministic Finite-State Automaton}
\acrodef{scltl}[scLTL]{syntactically co-safe LTL}
\acrodef{ltl}[LTL]{Linear Temporal Logic}
\newcommand{\acc}{\mathfrak{F}}
\newcommand{\type}{\mathsf{type}}

\newcommand{\mask}{\mathsf{mask}}

\newcommand{\decoys}{\mathsf{d}}
\newcommand{\hyperts}{\mathsf{HTS}}

\newcommand{\turn}{\mathsf{t}}

\newcommand{\goals}{\mathsf{t}}

\newcommand{\asurewin}{\mathsf{asw}} 

\newcommand{\nwstate}{\mathsf{NW}}

\newcommand{\nwstates}{\mathsf{NWs}}
\newcommand{\hosts}{\mathsf{Hosts}}
\newcommand{\servs}{\mathsf{Servs}}
\newcommand{\serv}{\mathsf{Serv}}

\newcommand{\creds}{\mathsf{Creds}}

\newcommand{\safe}{\mathsf{safe}}
\newcommand{\cosafe}{\mathsf{cosafe}}

\newcommand{\level}{\mathsf{level}}
\newcommand{\surewin}{\mathsf{sw}}
\newcommand{\commspec}{\neg \varphi_2}
\newcommand{\hidspec}{\psi}

\title{Secure-by-synthesis network with active deception and temporal logic specifications\\
  \thanks{This material is based upon work supported by the Defense Advanced Research Projects Agency (DARPA) under Agreement No. HR00111990015.}
  \thanks{\IEEEauthorrefmark{2} Huan Luo is a visiting student with Dr. Jie Fu at the Worcester Polytechnic Institute from Sept to Nov, 2019.}
}
\author{
\IEEEauthorblockN{Jie Fu\IEEEauthorrefmark{1}, Abhishek N. Kulkarni \IEEEauthorrefmark{1}, Huan Luo\IEEEauthorrefmark{2}, Nandi O. Leslie
   \IEEEauthorrefmark{3},
   and Charles A. Kamhoua 
 \IEEEauthorrefmark{3}}
\IEEEauthorblockA{\IEEEauthorrefmark{1} \IEEEauthorrefmark{2}Dept. of Electrical and Computer Engineering,\\
Robotics Engineering Program,\\
Worcester Polytechnic Institute, MA, US\\
 \IEEEauthorrefmark{1}\emph{jfu2, ankulkarni@wpi.edu},
\IEEEauthorrefmark{2}\emph{hluo12@126.com}}
  \IEEEauthorblockA{\IEEEauthorrefmark{3}  U.S. Army Research Laboratory, MD, US\\
  \emph{charles.a.kamhoua.civ, nandi.o.leslie.ctr@mail.mil}}
}

\begin{document}
\maketitle
\thispagestyle{plain}
\pagestyle{plain}
 
\nomenclature[1]{$G$}{Transition system for turn-based game arena.}
\nomenclature[1]{$S_i$}{Player $i$'s states in game arena.}
\nomenclature[1]{$A_i$}{Player $i$'s actions in game arena.}
\nomenclature[1]{$T$}{Transition function in the arena.}
\nomenclature[2]{$\calAP$}{A set of atomic propositions.}

\nomenclature[1]{$L_i$}{Player $i$'s perceived labeling function.}
\nomenclature[3]{$\pi^\surewin_i$}{Player $i$'s sure-winning strategy.}

\nomenclature[3]{$\pi^\asurewin_i$}{Player $i$'s almost-sure-winning strategy.}
\nomenclature[2]{$\Sigma=2^\calAP$}{Powerset of atomic propositions, also called the alphabet.}
\nomenclature[2]{$\Sigma^\omega$}{A set of words with infinitely many symbols in $\Sigma$.}
\nomenclature[2]{$\Sigma^\ast$}{A set of words with finitely many symbols in $\Sigma$.}
\nomenclature[4]{$\hgame^k$}{level-$k$ Hypergame.}
\nomenclature[4]{$\hyperts$}{Hypergame transition system.}
\nomenclature[5]{$s$}{A state in the game arena.}
\nomenclature[5]{$q$}{A state in an \ac{dfa}.}
\nomenclature[5]{$v$}{A state in the hypergame transition system.}
\nomenclature[3]{$\win_i^j$}{A set of winning states of player $i$ perceived by player $j$.}
\nomenclature[4]{$\hyperts\slash_{\pi_i}$}{Policy $\pi_i$-induced subgame.}

\begin{abstract}
This paper is concerned with the synthesis of strategies in network systems with active cyber deception. Active deception in a network employs decoy systems and other defenses to conduct defensive planning against the intrusion of malicious attackers who have been confirmed by sensing systems. In this setting, the defender's objective is to ensure the satisfaction of security properties specified in temporal logic formulas. We formulate the problem of deceptive planning with decoy systems and other defenses as a two-player games with asymmetrical information and Boolean payoffs in temporal logic. We use level-2 hypergame with temporal logic objectives to capture the incomplete/incorrect knowledge of the attacker about the network system as a payoff misperception. The true payoff function is private information of the defender. Then, we extend the solution concepts of $\omega$-regular games to analyze the attacker's rational strategy given her incomplete information. By generalizing the solution of level-2 hypergame in the normal form to extensive form, we extend the solutions of games with safe temporal logic objectives to decide whether the defender can ensure security properties to be satisfied with probability one, given any possible strategy that is perceived to be rational by the attacker. Further, we use the solution of games with co-safe (reachability) temporal logic objectives to determine whether the defender can engage the attacker, by directing the attacker to a high-fidelity honeypot. 
The effectiveness of the proposed synthesis methods is illustrated with synthetic network systems with honeypots. 
\end{abstract}

  \printnomenclature

\section{Introduction} 
In networked systems, many vulnerabilities may remain in the network even after being discovered, due to the delay in applying software patches and the costs associated with removing them. In the presence of such vulnerabilities, it is critical to design network defense strategies that ensure the security of the network system with respect to complex, high-level security properties. 
In this paper, we consider the problem of automatically synthesizing defense strategies capable of active deception that satisfy the given  security properties with probability one. We represent the security properties using temporal logic, which is a rich class of formal languages to specify correct behaviors in a dynamic system \cite{manna2012temporal}. 
For instance, the property, ``the privilege of an intruder on a given host is \textit{always} lower than the root privilege,'' is a safety property that asserts that the property must be true at all times. The property, ``it is the \textit{always} the case that \textit{eventually} all critical hosts will be visited,'' is a liveness property, which states that something good will always eventually happen. 

%An interesting result states that all temporal logic properties can be written as safety,  liveness, or a conjunction of safety and liveness properties \cite{Alpern1985}.

In the past, formal verification, also known as model checking, has been employed to verify the security properties of network systems, expressed in temporal logic \cite{Jha2002,ning2004techniques,Ramos2017}. 
These approaches construct a transition system that captures all possible exploitation by (malicious or legitimate) users in a given network.  
Then, a verification/model checking tool is used to generate an attack graph as a compact representation of all possible executions that an attacker can exploit to violate safety and other critical properties of the system. By construction,  the attack graph captures multi-step attacks that may exploit not only an individual vulnerability but also multiple vulnerabilities and  their causal dependencies.  Using an attack graph, the system administrator can perform analysis of the risks either offline or at run-time. % \cite{munoz-gonzalezEfficientAttackGraph2017}.  

However,  verification and  risk analysis with attack graphs have limitations: they do not take into account the possible defense mechanisms that can be used online by a security system during an active attack. For example, once an attacker is detected, the system administrators may change the network topology online (using software-defined networking) \cite{liu2016leveraging} or activate decoy systems and files.  Given increasingly advanced cyber attacks and defense mechanisms, it is desirable to synthesize, from system specification, the defense strategies that can be deployed online against active and progressive attacks to ensure a provably secured network. 
Motivated by this need, we present a game-theoretic approach to synthesize reactive defense strategies with active deception. Active deception employs decoy systems and files in synthesizing proactive security strategies, assuming a malicious attacker has been detected by the sensing system \cite{underbrink2016effective}.  

% In literature, the majority of work in game-theoretic planning with cyber-deception models the attacker's and defender’s objectives using a reward (or loss) functions and computes the strategy that maximizes (or minimizes) the objective function \cite{Jajodia2016,Cohen2006,ijcai2019-50}. 
% This is in contrast to specifying properties in temporal logic, where players get payoff only for satisfying the specification.
% In \cite{Hor2012}, dynamic Bayesian games  are employed to solve for defense strategies with active deception. The attacker is assumed to have incomplete information and, therefore, forms a belief about the defender’s unit. The players employ Bayesian rules to update their belief about the state in the game. Leveraging the attacker’s incomplete information, the defender may mislead the attacker’s belief and, thus, his actions to minimize the damage to the network measured by a state-dependent loss function. These approaches that use reward/loss functions are called quantitative. Qualitative analysis with attacker graphs has been studied extensively for verification but not so for planning. In qualitative reasoning, the goal is to ensure that, with probability one (almost-surely), certain properties in temporal logic are satisfied in the network during the dynamic interactions between the defender and the attacker.

Game theory has been developed to model and analyze defense deception in network security (see a recent survey in \cite{pawlickGametheoreticTaxonomySurvey2019}). Common models of games have been used in security include Stackelberg games, Nash games, Bayesian games, in which reward or loss functions are introduced to model the payoffs of the attacker and the defender. Given the reward (resp. loss) functions, the player's strategies are computed by maximizing (resp. minimizing) the objective function \cite{Jajodia2016,Cohen2006,ijcai2019-50}. 
Carrol and Grosu \cite{carrollGameTheoreticInvestigation2009} used a signaling game to study honeypot deception. The defender can disguise honeypots as real systems and real systems as honeypots. The attacker can waste additional resources to exploits honeypots or to determine whether a
system is a true honeypot or not. The solution of perfect Bayesian equilibrium provides the defend/attack strategies. Huang and Zhu \cite{al-shaerDynamicBayesianGames2019} used dynamic Bayesian games to solve for defense strategies with active deception. They considered both one-sided incomplete information, where the defender has incomplete information about the type of the attacker (legitimate user or adversary), and two-sided incomplete information, where the attacker also is uncertain about the type of the defender (high-security awareness or low-security awareness). Based on the analysis of  Nash equilibrium, the method enables the prediction of the attacker's strategy and proactive defense strategy to mitigate losses.

Comparing to the existing quantitative game-theoretic approach,  game-theoretic modeling, and qualitative analysis with Boolean security properties have not been developed for deception and network security. The major difference between quantitative and qualitative analysis in games lies in the definitions of the payoff function. Instead of minimizing loss/maximizing rewards studied in prior work,   the goal of qualitative reasoning is to synthesis a security policy that ensures, with probability one (\ie, almost-surely),  given security properties are satisfied in the network during the dynamic interactions between the defender and the attacker.

To synthesize provably secured networks with honeypot deception, we develop a game-theoretic model called ``$\omega$-regular hypergame'', played between two players: the defender and the attacker. A hypergame is a game of games, where the players play their individual perceptual game, constructed using the information available to each player. Similar to \cite{al-shaerDynamicBayesianGames2019}, we assume that the attacker has incomplete information about the game and is not aware of the deployment of decoy systems. However, both the defender and attacker are aware of each other's actions and temporal logic objectives (also known as $\omega$-regular objectives). Therefore, the defender and the attacker play different games, which together define the $\omega$-regular hypergame. 

To solve for active deception strategies in the $\omega$-regular hypergame, we first construct the game graphs of the individual games being played by the attacker and defender. In this context, a game is defined using three components: (a) a transition system, called arena, which captures all possible interactions over multiple stages of the game; (b) a labeling function that relates an outcome--a sequence of states in the game graph--to properties specified in logic; and (c) the temporal logic specifications as the players' Boolean objectives. The construction of the transition system is closely related to the attack graph as it  captures the causal dependency between the vulnerabilities in the network. The only difference is that in the attack graph analysis, the transitions are introduced by the attacker's moves only, whereas in the game transition system the transitions may be triggered by both attacker's and the defender's actions. 
Cyber deception is introduced through payoff manipulation: when the attacker does not know which hosts are decoys,   the attacker might misperceive herself to be winning if the security properties of the defender have been violated, when in fact, they are not. 
We define a hypergame transition system for the defender to synthesize 1) the rational, winning strategy of the attacker given the attacker’s (mis)perception of the game; and then 2) the deceptive security strategies for the defender that exploits the attacker's perceived winning strategy. 

The paper is structured as follows. In Sec.~\ref{sec:prelim}, we present the definition of $\omega$-regular games, and show how such a game can be constructed from a network system. In Sec.~\ref{sec:hypergame-modeling}, we formulate a modeling framework called ``cyber-deception $\omega$-regular game'' and show how to capture the use of decoys as a payoff manipulation mechanism in such a game. In Sec.~\ref{sec:synthesis}, we present the solution for the cyber-deception $\omega$-regular game with asymmetrical information. Using the solution of the game, a defender's strategy, if one exists, can be synthesized to ensure that the security properties are satisfied with probability one. This strategy uses both active deception and reactive defense mechanisms. Further, we analyze whether a strategy exists to ensure that the defender can achieve a preferred outcome, for example, forcing the attacker to visit a honeypot eventually. Finally, we use examples to illustrate the methods and effectiveness of the synthesized security strategy. In Sec.~\ref{sec:conclude} we conclude and discuss potential future directions.

 %Further, the defender can leverage the weakness and mistakes in the perceived winning strategies of a rational attacker for improving security and safety of the network system. In Sec.~\ref{subsec:permissive}, we introduce a synthesis algorithm to generate the maximally permissive strategy for the attacker as a compact representation of all randomized strategies that the attacker can employ in her peceptual game. In Sec.~\ref{subsec:deceptive_win}, we show that the permissive strategy of the attacker induces a subgame from the original game, for which reactive synthesis algorithms can be applied to derive sure-winning strategy for the defender given his two objectives: Ensuring the safety of the network, and luring the attacker to reach certain high-fidelity honeypots (for understanding the motives of the attacker).
 
%  In Sec~\ref{sec:experiment} we include the experiment analysis in a synthetic network system with cyberdeception.

\section{Preliminaries and Problem Formulation}
\label{sec:prelim}
 
 \noindent \textbf{Notation}: Given a set $X$, the set of all possible
 distributions over $X$ is denoted $\dist{X}$. For a finite set $X$,
 the powerset (set of subsets) of $X$ is denoted $2^X$. For any
 distribution $d\in \dist{X}$, the support of $d$, denoted $\supp(d)$,
 is the set of elements in $X$ that has a nonzero probability to be
 selected by the distribution, \ie, $\supp(d)=\{x\in X\mid
 d(x)>0\}$. Let $\Sigma$ be an alphabet, a sequence of symbols
 $w=\sigma_0 \sigma_1 \ldots \sigma_n $, where $\sigma_i \in \Sigma$,
 is called a \emph{finite word} and $\Sigma^\ast$ is the set of finite
 words that can be generated with alphabet $\Sigma$.  We denote
 $\Sigma^\omega$ the set of words obtained by concatenating the
 elements in $\Sigma$ infinitely many times. Notations used in this paper can be found in the nomenclature.

  \subsection{The game arena of cyber-deception game}
 Attack graph is a formalism for automated network security analysis. Given a network system, its mathematical model can be constructed as a finite-state transition system which includes a set of states describing various network conditions, the set of available actions that can be performed by the defender or the attacker to change the network conditions, their own states, and transition function that
 captures the pre- and post-conditions of actions given states. A pre-condition is a logical formula that needs to be satisfied for an action to be taken. A post-condition is a logical formula describes the effect of an action in the network system.  Various approaches to attack graph generation have been proposed (see a recent survey  \cite{AitMaalemLahcen2018}). 

 The attacker takes actions to exploit the network, such as remote
 control exploit, escalate privileges, and stop/start services on a
 host under attack.  When generating the attack graph, we also
 introduce a set of defender's actions, enabled by
 defensive software such as firewalls, multi-factor user
 authentication, software-defined networking to remove some
 vulnerabilities in real-time. After incorporating both defender's and
 attacker's moves, we obtain a two-player game arena in the form of a
 deterministic transition system in Def.~\ref{def:arena}. In this
 game, we refer the defender to be player 1, P1 (pronoun `he') and
 the attacker to be player 2, P2 (pronoun `she').
\begin{definition}
\label{def:arena}
A turn-based game arena consists of a tuple 
\[ G = \langle S,A, T , \calAP,  L \rangle,\] where:
\begin{itemize}
\item $S=S_1\cup S_2$ is a finite set of states partitioned into P1's states $S_1$ and P2's states $S_2$;
\item $A_1$ (resp., $A_2$) is the set of actions for P1 (resp., P2); 
\item $T : (S_1 \times A_1)\cup (S_2 \times A_2) \rightarrow S$ is a \emph{deterministic} transition function that maps  a state-action pair to a next state.
\item $\calAP$ is the set of atomic propositions.
\item $L: S\rightarrow 2^\calAP$ is the labeling function that maps
  each state $s\in S$ to a set $L(s)\subseteq \calAP$ of atomic
  propositions evaluated true at that state.
 \end{itemize}
\end{definition}
The set of atomic propositions and labeling function together enable
us to specify the security properties using logical formulas.  A
\emph{path} $\rho=s_0s_1\ldots $ in the game arena is a
(finite/infinite) sequence of states such that for any $i\ge 0$, there
exists $a\in A_1\cup A_2$, $\delta(s_i,a) = s_{i+1}$.
A path
$\rho=s_0s_1\ldots $ can be mapped to a word in $2^\calAP$,
$w=L(s_0)L(s_1)\ldots $, which is  evaluated against the
pre-defined security properties in logic.

In this work, we consider deterministic, turn-based game arena, \ie,  at any step, either the attacker or the defender takes an action and the outcome of that action is deterministic. This turn-based interaction can be understood as if the defense is \emph{reactive} against attacker's exploits. The turn-based interaction has been  adapted in cyber-security research \cite{chakraborty2018hybrid}. The generalization to concurrent stochastic game  arena is a part of our future work.

The model of  game arena is generic enough to capture many attack  graphs generated by different approaches. For example, consider the lateral movement attack, the state $s$ can include information about the source host of the attacker, the user credential obtained by the attacker, and the current network condition--the active hosts and services. Depending on the pre- and post- conditions of each vulnerability, the attacker can select a vulnerability to exploit.   For example, a vulnerability named ``IIS\_overflow'' requires the attacker to have a credential of user or a root, and the host running IIS webserver \cite{Sheyner2004}. After exploiting the vulnerability, the attacker moves from the source host to a target host, changes her credential and the network condition. For example, after the attacker exploits the vulnerability of ``IIS\_overflow'', the service will be stopped on the target host and the attacker gains the credential as a root user. Then the defender may choose to stop services on a host or to change the network connectivity, or to replace a current host with a decoy as defense actions. These defense actions again change the network status--resulting in a transition in the arena.
%
%In this work, we generate an \emph{attacker graph} using multiple-prerequisite graphs \cite{}. In a multiple-prerequisite attack graph
%   consists of three types of nodes, i.e. user state nodes, prerequisite nodes, and vulnerability instance nodes. 
%   A user state node represents the level of access on a given host. 
%    Prerequisite nodes represent   preconditions of one or several attacks. Vulnerability instance nodes represent particular vulnerabilities. Directed edges from state nodes to prerequisite nodes represent the capabilities those states enable for the attacker. Prerequisite nodes point to vulnerability instance nodes that represent the set of attacks that the prerequisite node enables. Directed edge from vulnerability instance nodes to a single state node represent the state the attacker can reach by successfully exploiting the vulnerability. 
    
    % An example of the multiple-prerequisite attack graph given in \cite{artz2002netspa} is shown in Fig.~\ref{fig:mp-graph}. REVISE}.
%    
%    \begin{figure}
%        \centering
%        \includegraphics[width=0.25\textwidth]{ex_mpg.png}
%        \caption{An example of multi-prerequisite attack graph  \cite{artz2002netspa}.}
%        \label{fig:mp_graph}
%    \end{figure}
%     

%  Then, with
%   asymetrical information about the components of the game, we have a
%   hypergame where players have different perceptual games, \ie, games
%   that they perceive to be.

\subsection{The payoffs in the game} 
We consider that the defender's objective is to satisfy  security properties of the system, specified in temporal logic. The attacker's objective is to violate security properties of the system. We assume the security properties are common knowledge between the attacker and the defender, corresponding to the worst case assumption of the attacker. 
Next, we give the formal syntax and semantics of \ac{ltl} and then several examples related to network security analysis.

Let $\calAP$ be a set of atomic propositions. Linear Temporal Logic
(\ac{ltl}) has the following syntax,
\[ \varphi := \top \mid \bot \mid p \mid \varphi \mid \neg\varphi \mid \varphi_1 \land \varphi_2 \mid \bigcirc \varphi \mid \varphi_1 {\until} \varphi_2, \] where 
\begin{itemize}
    \item $\top,\bot$ are universally true and false, respectively.
    \item $p \in \calAP$ is an atomic proposition.
    \item $\bigcirc$  is a temporal operator called the ``next'' operator (see semantics below).
    \item $\until$ is a temporal operator called the ``until'' operator (see semantics below).
    \end{itemize}
    
% $\top,\bot$ are universally true and false, respectively.n, $\bigcirc, \until$,  $\Eventually $, and $\Always$ denote the temporal modal operators for \textit{next}, \textit{until}, \textit{eventually} and \textit{always}.

Let $\Sigma\coloneqq 2^\calAP$ be the finite alphabet. 
Given a word $w\in \Sigma^\omega$, let $w[i]$ be the $i$-th element in the word and $w[i\ldots]$ be the subsequence of $w$ starting from the $i$-th element. For example, $w=abc$, $w[0]=a$ and $w[1\ldots] = bc$. Formally,
we have the following definition of the semantics: 
\begin{itemize}
    \item $w\models p$ if $p \in w[0]$;
    \item $w\models \neg p $ if $p \notin w[0]$;
    \item $w\models \varphi_1\land \varphi_2$ if $w \models \varphi_1$ and $w\models \varphi_2$.
    \item $w\models \bigcirc \varphi$ if $w[1\ldots] \models \varphi$.
    \item $w\models \varphi \until \psi$ if $\exists i \ge 0$, $w[i\ldots] \models \psi$ and $\forall 0\le j<i$, $w[j\ldots]\models \varphi$.  
\end{itemize}
From these two temporal operators ($\bigcirc, \until$), we define two
additional temporal operators: $\Eventually$ eventually and $\Always$
always. Formally, $\Eventually \varphi$ means $\truev \until
\varphi$. $\Always \varphi$ means $\neg \Eventually \neg \varphi$.
For details about the syntax and semantics of \ac{ltl}, the readers
are referred to \cite{Pnueli1989}.

Next, we present some examples of \ac{ltl} formulas for describing security properties in a network. Consider a set of atomic propositions $\calAP = \{p_1,p_2,p_3,p_4\}$, where
\begin{itemize}
    \item $p_1$: service A is enabled on host 1.
    \item $p_2$: the attacker has the root privilege on a host.
    \item $p_3$:  the attacker is on host $1$.
    \item $p_4$: the attacker is on host $2$.
\end{itemize}
Using the set of atomic propositions, the following security properties or attacker's objectives can be described:
\begin{itemize}
    \item   $\Always p_1$:  ``Service A will always be enabled on host 1''.
    \item $\Always \Eventually p_1$: ``It is always the case that service A will be eventually enabled on host 1.''
    \item  $\Eventually (p_3 \land p_2)$:  ``Eventually the attacker reaches host $1$ with a root privilege''.
    \item  $(p_1\land p_3\land p_2)\implies \bigcirc \neg p_1$: ``If the attacker is at host $1$ with a root access and service A is enabled, then at the next step she will stop the service A on host 1.''
    \item   $\Eventually(p_3\land \Eventually p_4)$: ``Eventually the attacker visits host $1$ and then visits host $2$.''
\end{itemize}
Using the semantics of \ac{ltl}, we can evaluate whether a word $w\in \Sigma^\omega$ satisfy a given formula. For example, $\emptyset \emptyset \{p_3,p_2\}$ satisfies the formula  $\Eventually (p_3 \land p_2)$.

%   \begin{remark}
%   The turn-based game arena can be derived  from the
%   \emph{attacker model} in \cite{Jha2002,Sheyner2002}. An attacker model  is
%   defined to be a finite-state transition system with a set $\calAP$
%   of atomic propositions and a labeling function. A state in the model
%   is a valuation of variables describing the attacker, the defender
%   and the system (detailed later). The transitions in the system correspond to actions
%   taken by an attacker which lead to a change in the overall state of
%   the system. The initial state denotes the state of the system where
%   no damage is made and the attacker has entered the system via an
%   entry point.  The only difference in the game arena and the attacker
%   model is that in the game arena there are two types of transitions:
%   Transitions taken by P1 (the defender) and transitions taken by P2
%   (the attacker). 
% \end{remark}

% In omega-regular games, temporal logic is used to specify objectives of players in the game arena. 

In this work, we restrict to a subclass of \ac{ltl} called
syntactically safe and co-safe \ac{ltl} \cite{kupferman2001model},
which are closely related to  bad and good prefixes of languages:
Given an \ac{ltl} formula $\varphi$, a \emph{bad} prefix for $\varphi$
is a finite word that can not be extended in any way to satisfy
$\varphi$; a good prefix for $\varphi$ is a finite word that can
extended in any way to satisfy $\varphi$.  \emph{Safe LTL} is a set of
\ac{ltl} formulas for which any   infinite word that does not satisfy the formula has a
finite bad prefix. \emph{co-safe LTL} is the set of LTL formulas for
which any satisfying infinite word has a finite good
prefix. % A co-safe formula can be written in
% the positive normal form where negation can only appear next to atomic
% propositions. Thus, co-safe \ac{ltl} does not allow the use of the
% ``always'' operator $\Always$. Safe LTL is the set of formulas for
% which only the $\Always$ and $\bigcirc$ temporal operators occur. Safe
% LTL and co-safe \ac{ltl} formulas are dual to each other--if $\varphi$
% is a safe \ac{ltl} formula, then $\neg \varphi$ is a co-safe \ac{ltl}
% formula, and vice versa.

The advantage of restricting to safe and co-safe \ac{ltl} is that both
types of formulas can be represented by \ac{dfa}s with different
acceptance conditions.  A \ac{dfa} is a tuple
$\calA = (Q, \Sigma, \delta, I, (F, \type))$ which includes a finite
set $Q$ of states, a finite set $\Sigma =2^\calAP$ of symbols, a
deterministic transition function
$\delta: Q\times \Sigma \rightarrow Q$, and a unique initial state
$I$. The acceptance condition is specified in a tuple $(F, \type)$
where $F\subseteq Q$ and $\type \in \{\safe, \cosafe\}$. The intuition
is that states in a \ac{dfa} for an \ac{ltl} formula are in fact
finite-memory states to keep track of partial satisfaction of the
said formula \cite{esparza2016ltl}.

Given a word $w = \sigma_0\sigma_1\ldots \in \Sigma^\omega$, its
corresponding \emph{run} in the \ac{dfa} is a sequence of automata
states $q_0,q_1,\ldots $ such that $q_0=I$ and
$q_{i+1} = \delta(q_i, \sigma_i)$ for $i\ge 0$.  Different types of
\ac{dfa}s define different accepting conditions:
\begin{itemize}
    \item  when $\type=\safe$. A word $w$ is \emph{accepted} if its corresponding run only visits  states in $F$. That is, for all $i\ge 0$, $q_i \in F$.
    \item when $\type = \cosafe$. A word $w$ is \emph{accepted} if its corresponding run  visits  a state in $F$. That is, there exists $i \ge 0$, $q_i \in F$.
    \end{itemize} 
    A safe \ac{ltl} formula $\varphi$ translates to a \ac{dfa} with
    $\type = \safe$. A co-safe \ac{ltl} formula $\varphi$ translates
    to a \ac{dfa} with $\type =\cosafe$. We will present several examples of \ac{ltl} formulas and their \ac{dfa}s in Section~\ref{sec:hypergame-modeling} and Section~\ref{sec:experiment}.
    
%  \begin{example}
% \todo[inline]{Todo: Given a small example. }
%  \end{example}
 \begin{remark}
The restriction to these two types of acceptance conditions does not allow us to specify recurrent properties, for example, ``always eventually a service is running on host $i$.'' However, we can specify temporally extended goals and a range of safety properties in a network systems. The extension to  more complex specifications is also possible but requires different synthesis algorithms that deal with recurrent properties.
\end{remark}

Putting together the game arena and the payoffs of players, we can formally define $\omega$-regular games.

\begin{definition}[$\omega$-regular game]
An $\omega$-regular game with safe/co-safe objective is $\game =(G, (\varphi_1,\varphi_2))$ that includes a game arena $G$, and player $i$'s objectives expressed by safe/co-safe \ac{ltl} formulas.
\end{definition}
An $\omega$-regular game $\game$ is zero-sum if $\varphi_1 = \neg \varphi_2$--that is, the formula that P1 wants to satisfy is the negation of P2's formula. 

In an $\omega$-regular game, a \emph{deterministic strategy} is a function
$\pi_i: S^\ast \rightarrow A_i$ that maps a history to an action. A
\emph{ set-deterministic strategy} is a function
$\pi_i: S^\ast \rightarrow 2^{A_i}$ that maps a history to a subset
of actions among which player $i$ can select nondeterministically. A
\emph{mixed/randomized strategy} is a function
$\pi_i: S^\ast\rightarrow \dist{A_i}$ that maps a history into a
distribution over actions. If the strategy depends on the current state only, then we call the
strategy \emph{memoryless}.  A strategy is \emph{almost-sure winning} for player $i$ if and only if by committing to this strategy, no matter which strategy the opponent commits to, the outcome of their interaction satisfies the objective $\varphi_i$ of player $i$, with probability one.

The following result is rephrased from \cite{buchi1969solving}.
\begin{theorem}
  \label{thm:determined} \cite{buchi1969solving} A zero-sum turn-based
  $\omega$-regular game is determined, \textit{i.e.} for a given history, only one
  player (with a finite memory strategy) can win the game. 
\end{theorem}

\section{Modeling: A hypergame for cyber-deception}
\label{sec:hypergame-modeling}
% \subsection{Games with asymmetrical incomplete information}

When decoy systems are employed, the game between the attacker and defender is a game with asymmetric, incomplete information. 
% A game with incomplete information means that at least one of the players has partial knowledge about the game arena, or the payoffs of other players, or both.
In such a game, at least one player has privileged information over
other players. In active cyber-deception with decoys, the defender has
correct information about honeypot locations. We employ hypergame,
also known as game of games, to capture the defender/attacker
interaction with asymmetric information.

\begin{definition}\cite{bennett1980hypergames,Vane}
  A level-1 two-player  hypergame is a pair
  \[ \hgame^1 = \langle \game_1, \game_2 \rangle, \] where
  $\game_1,\game_2$ are games perceived by players P1 and P2,
  respectively.  A level-2 two-player hypergame is a pair
  \[\hgame^2= \langle \hgame^1, \game_2 \rangle\] where P1 perceives the
  interaction as a level-1 hypergame and P2 perceives the interaction
  as game $\game_2$.
\end{definition}
    In general, if P1 computes his strategy
  using $m$-th level hypergame and P2 computes her strategy using
  an $n$-th level hypergame with  $n < m$, then the resulting hypergame is said
  to be a level-$m$ hypergame given as
  \[\hgame^m = \langle \hgame^{m-1}_1, \hgame^n_2 \rangle. \]

% The definition of two-player hypergames generalizes to
%   $k$ players for $k>2$: a $m$-level $k$-player hypergame is that at
%   least one player is playing a $m-1$ level hypergame. The rest of players
%   play  hypergames with potentially different levels and the levels of the hypergame are less than $m-1$. 
  We refer to  the game perceived by player $i$ as the \emph{perceptual game}
  of player $i$.

  In active cyber-deception, the defender uses decoys to induce one-sided misperception  of attacker. We introduce the following
  function, called \emph{mask}, to model the resulting misperception
  of the attacker:
\begin{definition}[Mask]
  Given the set of symbols $\Sigma$, P2's perception of $\Sigma$ is given by a
  \emph{parameterized mask} function
  $\mask: \Sigma\times \Theta \rightarrow \Sigma$ where for each
  $\sigma \in \Sigma$ and a given parameter $\theta\in \Theta$, P2
  perceives $\mask(\sigma,\theta)$. We say that two symbols $\sigma_1, \sigma_2 \in \Sigma$ are
  \emph{observation-equivalent} for P2, if
  $\mask(\sigma_1,\theta) = \mask(\sigma_2,\theta)$.  Let
  $[\sigma]_{\theta} = \{\sigma' \in \Sigma\mid \mask(\sigma,\theta) =
  \mask(\sigma',\theta)\} $. 
\end{definition}
Note that the set of observation-equivalent states partitions
  the set $\Sigma = \bigcup_{\sigma\in \Sigma} [\sigma]_{\theta}$ due to transitivity in this equivalence relation. We use a simple example to illustrate this definition.
  
\begin{example}Suppose there is an atomic proposition $p_5$: ``host
  $1$ is a decoy.'' and P2 cannot observe the valuation of the
  proposition, then P2 is unable to distinguish a regular host from a
  decoy host. The parameter $\theta$ can be a set of hosts in the
  network where decoys are placed.
\end{example}
  
The inclusion of parameter $\theta$ will allow us to define different
misperceptions.  Given P1's labeling function and
P2's labeling function $L_2: S\rightarrow \Sigma$, the misperception
of P2 given the labeling function can be represented as
$L_2 (s) = \mask(L_1(s),\theta)$, for each $s\in S$. Slightly abusing
the notation, we denote the function $L_2 = \mask(L_1, \theta)$. In
this paper, we consider a fixed mask function, which means the
parameter $\theta$ is fixed. Thus, we omit $\theta$ from the mask
function. It is noted that the optimal selection of $\theta$ 
is the problem of mechanism design of the hypergame (c.f. \cite{Sandholm99distributedrational}), which is beyond the scope of this paper.

\begin{definition}
\label{def:level2}
Given the mask function $\mask$, the interaction between players is a level-2 hypergame in which
  P1 has complete information about the labeling function $L_1= L:S\rightarrow 2^\calAP$ in the arena
  $G$ and P2 has a misperceived labeling function in  
  $G$. In addition, P1 knows P2's misperceived labeling function. The level-2 hypergame is a tuple
  
\[
\hgame^2 = \langle \hgame^1, \game_2\rangle,
\]
where $\hgame^1 = \langle \game_1,\game_2\rangle$ is a level-1
hypergame with
\[
\game_1 = \langle G_1 = \langle S, A, T, \calAP, L_1 \rangle, (\varphi_1, \varphi_2) \rangle
\]
and 
\[ \game_2= \langle G_2 = \langle S, A, T, \calAP, L_2 = \mask(L_1,
  \theta) \ne L\rangle, (\varphi_1,\varphi_2)\rangle,
\]
where $\varphi_i$ is player $i$'s  objective in temporal logic, for $i=1,2$.
\end{definition}

In this work, we consider the following payoffs for players.

\begin{itemize}
\item The attacker's objective is given by a co-safe \ac{ltl}  formula $\varphi_2$. For example, $\varphi_2$ can be ``eventually visit host 1.''
\item The defender's objective is given by a preference
  $\commspec \land \hidspec \succ \commspec$ where $\hidspec$ is a
  co-safe \ac{ltl} formula and $\succ$ is the operator for ``is strictly preferred to.'' That is, P1 prefers to prevent P2 from
  achieving her objective \emph{and} to satisfy a hidden objective. If
  not feasible, then P1 is to satisfy the objective $\commspec$, which
  is the negation of P2's specification and known to P2. 
  \end{itemize}
  
 For example, P1's objective $\varphi_1$ could be ``always stop attacker from reaching host 1'', while the additional objective $\psi$ could be ``eventually force the attacker to visit a decoy''.
% two cases:
% \begin{itemize}
%     \item Case I: $\varphi_1 ,\varphi_2$ are common knowledge to both players and $\varphi_1 = \neg \varphi_2$.
%     \item Case II: $\varphi_1 = \commspec \land \hidspec$ and
%       $\varphi_2 = \neg \commspec$. P1's objective is a
%       conjunction of two subformulas: The formula $\commspec$ is
%       the common knowledge to both players; and the formula
%       $\hidspec$ is the hidden objective of player 1, unknown to
%       player 2. P1's payoff is given by a preference $\succ$ such that
%       $\varphi_1\succ \commspec$. That is, P1 prefers to satisfy
%       the objective known to P2 \emph{and} the hidden objective. If
%       not feasible, then P1 is to satisfy the objective
%       $\commspec$, which is known to P2.
% \end{itemize}
% Note that Case I is indeed a special case of Case II, where P1's hidden objective is $\truev$ for universally true.

In level-2 hypergame, P1's strategy is influenced by his perception of
P2's perceptual equilibrium. Since P2 plays a level-0 hypergame, her
 perceptual equilibrium is the solution of $\game_2$. Player 1 should leverage the weakness in player 2's strategy to achieve better outcomes concerning his objective. To this end, we aim to solve the following qualitative planning problem with cyber-deception using decoys.
\begin{problem}
  Given a level-2 hypergame in Def.~\ref{def:level2}, synthesize a strategy for P1, if exists, such that \emph{no matter which} equilibrium P2 adopts in her perceptual game, P1 can ensure to satisfy his most preferred logical objective \emph{with probability one}, \ie, \emph{almost surely.}
\end{problem}

\section{Synthesis of deceptive strategies}
\label{sec:synthesis}
In this section, we present the synthesis algorithm for deceptive strategies. First, we show that when temporal logic specifications are considered, the defender requires finite memory to monitor the history (a state sequence) with respect to the partial satisfaction of given defender's and attacker's objectives. This construction of finite-memory states and transitions between these states is given in Sec.~\ref{sec:finite-memory}. Second, we construct a hypergame transition system in Sec.~\ref{sec:hypergame} with which the planning problem for the defender reduces to solving games with safe and co-safe objectives. Third, we compute the set of rational strategies of the attacker given her perceptual game in Sec.~\ref{sec:perceptualgame-P2}. Lastly in Sec.~\ref{subsec:deceptive_win}, we show how to synthesize the deceptive winning strategy for the defender, assuming that the attacker commits to an arbitrary strategy perceived to be rational by herself. The complexity analysis is given in Sec.~\ref{subsec:complexity}.
% We study a class of security problems in which the objective of the
% attacker is a syntactically cosafe LTL formula $\varphi_2$, which is
% equivalently represented as a \ac{dfa}
% $\calA_2 = (Q_2,2^{\calAP_2}, \delta_2, I_2, (F_2,\cosafe))$. Further,
% we assume that $\calA_2$ is complete, \ie, $\delta_2(q, \sigma)$ is
% defined for any $q\in Q_2$ and $\sigma\in \Sigma_2$. An incomplete
% automaton can be made complete by adding a nonaccepting sink state
% $\sink$ and let $\delta_2(q,\sigma)=\sink$ if the function $\delta_2$
% is previously undefined for the state-action pair $(q,\sigma)$.

% Because
% $\varphi_2$ is a cosafe formula, given the assumption that the
% attacker's objective is the common knowledge, $\commspec$ is a
% safe \ac{ltl} formula such that $ \commspec=\neg \varphi_2$.
% Using the duality of safe and cosafe formulas, the \ac{dfa} for $\commspec$, denoted $\calA_{1,c}$ shares the same states, alphabet, and transitions as in $\calA_2$ but different acceptance condition:
% \[
%   \calA_{1,c} = \langle Q_2,2^{\calAP_2}, \delta_2, I_2, (Q_2\setminus F_2,\safe) \rangle, \rangle 
% \]
% where the acceptance condition is cosafe condition.

% In addition, we assume that the hidden objective $\hidspec$ of P1 is
% a cosafe \ac{ltl} formula. This hidden objecitive is represented by a cosafe \ac{dfa} $\calA_{1,h} = \langle Q_1,2^{\calAP_1}, \delta_1, I_1, (F_1, \cosafe)\rangle$. We assume that the \ac{dfa} $\calA_{1,h}$ is also complete.

\subsection{Monitoring the history with finite memory of P1}
\label{sec:finite-memory} 
To design a deceptive strategy, P1 needs to monitor the
history of states with respect to both players' objectives--that is, maintaining the evolution of some finite-memory states. Thus, we first introduce  a
product using the \ac{dfa}s for $\hidspec$ and $\varphi_2$. The states of this product constitutes this set of ``finite-memory states''.  Later, we use an example to illustrate this construction. 

\begin{definition}
  \label{def:product}
  Given two complete\footnote{A \ac{dfa}
    $\calA = \langle Q,\Sigma, \delta, I, (F,\cosafe) \rangle $ is
    complete if for any symbol $\sigma\in \Sigma$, for any state
    $q\in Q$, $\delta(q,\sigma) $ is defined. An incomplete \ac{dfa}
    can always be made complete by adding a non-accepting sink state
    and redirecting all undefined transitions to the sink. } \ac{dfa}s,
  $\calA_{1}= \langle Q_1,\Sigma, \delta_1, I_1, (F_1,
  \cosafe)\rangle$--the defender's hidden co-safe  \ac{ltl} objective and
  $\calA_2 = \langle Q_2,\Sigma, \delta_2, I_2,
  (F_2,\cosafe))\rangle$--the attacker's  co-safe \ac{ltl} objective
  $\varphi_2$ and the mask function $\mask:\Sigma \rightarrow \Sigma$, the
  \emph{product automaton} given P2's misperception  $\calA_1 \otimes \calA_2$ is a \ac{dfa}:
  \[
    \calA =\langle Q, \Sigma, \delta, I,(\calF_{1}, \cosafe), (\calF_2,\cosafe)\rangle,
  \]
  where:
  \begin{itemize}
  \item $Q = Q_1\times Q_2$ is the state space.
    \item $\Sigma$ is the  alphabet.
    \item $\delta: Q\times \Sigma \rightarrow Q$ is defined as
      follows: Let $(q_1,q_2), (q_1',q_2') \in Q$,
      $\delta((q_1,q_2),\sigma)= (q_1',q_2')$ if and only if
      \begin{itemize}
          \item   $\delta_1(q_1,\sigma )=q_1'$  and
          \item  there exists $\sigma'$ such
      that $\delta_2(q_2, \sigma') = q_2'$ and
      $\mask(\sigma)= \mask(\sigma')$--that is, $\sigma'$ is observation-equivalent to $\sigma$ from P2's viewpoint.
      \end{itemize}
    %   \item if $\sigma \in \Sigma_1 \setminus \Sigma_2$, $q_1'= \delta_1(q_1,\sigma)$ and $q_2'=q_2$;
    %    \item if $\sigma \in \Sigma_2\setminus \Sigma_1$, $q_2' = \delta_2(q_2,\sigma)$ and $q_1'= q_1$;
            \item $I=  (I_1,I_2)$.
   \item $(\calF_1 ,\cosafe)= F_1\times Q_2$.
  \item $(\calF_{2},\cosafe)= Q_1\times F_2 $. %\setminus \calF_{1,c}= Q_1\times F_2$.
  \end{itemize}
  \end{definition}
  The product computes the transition function using the union
  of \ac{dfa}s $\calA_1$ and
  $\calA_2$, while considering the observation-equivalent classes of symbols under the mask function. It maintains two acceptance conditions for accepting the languages for 
  $\calA_1$ and $\calA_2$, respectively.
  In fact, given the first acceptance condition, $\calA = \langle Q, \Sigma, \delta, I,(\calF_{1}, \cosafe) \rangle $ accepts the same language as \ac{dfa} $\calA_1$. Given the second acceptance condition,  $\calA = \langle Q, \Sigma, \delta, I,(\calF_{2}, \cosafe) \rangle $ 
  accepts a set $L$ of words such that $w=\sigma_0\sigma_1\ldots \sigma_n \in L$ if there exists a word $w'  =\sigma_0'\sigma_1'\ldots \sigma_n'$,  where $\mask(\sigma_i)=\mask(\sigma_i')$ for $i=0,\ldots, n$, that is accepted by \ac{dfa} $\calA_2$.

  It can be proven that the transition is deterministic. 
  
  \begin{lemma}
  \label{lma:determinism}
  If $\delta((q_1,q_2),\sigma)=(q_1',q_2')$ is defined, then there exists only one $\sigma'$ such that $\delta(q_2,\sigma')=q_2'$ and $\mask(\sigma)=\mask(\sigma')$.
  \end{lemma}
  The proof is included in Appendix.
  
  It is noted that a state in a \ac{dfa} captures a subset of sub-formulas that have been satisfied given the input to reach the state from the initial state \cite{esparza2016ltl}.   Intuitively, given an input word $w$,  $\delta((I_1,I_2),w) = (q_1,q_2)$ captures 1)   P1's knowledge about the  true logical properties satisfied by reading the input; 2) P1's knowledge about what logical properties that P2 thinks have been satisfied by reading the input. Next, we use examples to illustrate the product definition.
 
  % Based on the preference of P1, a word $w$ accepted by $\calA$ with
  % the acceptance condition $\acc_{1,c}$ is preferred to one accepted
  % by $\calA$ with the acceptance condition $\acc_{1,h}$.

\begin{example} 
\label{ex:simple-specs} Consider the following example of players' objectives:
\begin{itemize}
    \item the  attacker's \emph{co-safe objective} is given by $\varphi_2\coloneqq \Eventually \goals$ for eventually reaching a set of goal states, labeled $\goals$,  in the game graph. A goal state can be that the attacker reaches a critical host.  
    \item the defender's \emph{objective} known to the attacker is
      $\commspec $, which is satisfied when the attacker is confined
      from visiting any goal state.
    \item  the defender's \emph{hidden, co-safe objective} is $ \hidspec \coloneqq \Eventually \decoys$, which is satisfied when the attacker reaches a honeypot labeled $\decoys$. In this co-safe objective, the defender is to lead the attacker into a honeypot. The attacker's activities at high-fidelity honeypots can be analyzed for understanding the intent and motives of the attacker.
\end{itemize} 
Given the task specification, we generate \ac{dfa}s $\calA_1, \calA_2$
in Fig.~\ref{fig:ex-specs}. The set $\calAP =\{\goals, \decoys\}$
where $\goals \coloneqq $ ``the current host is a target"; and
$\decoys \coloneqq$ ``the current host is a decoy''. The alphabet
$\Sigma = 2^\calAP = \{\emptyset, \{\decoys\}, \{\goals\}, \{\decoys,
\goals\}\}$. The mask function is such that
$\mask(\{\decoys\}) = \emptyset$ and $\mask(\{\decoys, \goals\})= \{\goals\}$ In words, P2 does not know which host
is a decoy. The symbol $\top$ stands for universally true. Transition $q\xrightarrow{\top}q'$
means $q\xrightarrow{\top}q'$ for any $\sigma\in \Sigma$.
The product
automaton is given in Fig.~\ref{fig:sync_product}.  For example, the
transition from $(0,0)\xrightarrow{\{\goals, \decoys\}} (1,1)$ is
defined because
$\delta_1(0,\{\goals, \decoys\}) = \delta_1(0,\{\decoys\})=1$ and
$\mask(\{\goals, \decoys\}\}) = \{\goals\}$ and
 $\delta_2(0,\{\goals\}) = 1$.  The reader can verify that $(0,0)\xrightarrow{\{\decoys\}} (1,0)$ is defined because $\mask(\{\decoys\})=\emptyset$.

 The
       defender's co-safe objective $ \calF_{1}= \{(1,0),(1,1)\}$
       and the attacker's co-safe objective
       $\calF_{2}=\{ (1,1),(0,1)\}$.
       If the state in $\calF_1$ is visited, then the defender knows that the attacker visited a decoy. The state $\{1,1\}$ lies in the intersection of $\calF_1$ and $\calF_2$. If this state is visited, then the defender knows that the attacker visited a decoy and the attacker wrongly thinks that she has reached a critical host. 
 %The transition
% $(0,1)\xrightarrow{\{\goals,\decoys\}} (1,1)$ is defined because
% $\delta_1(0,\{\goals, \decoys\})=1$ and $ \delta_2(1,\{\goals\})=1$
% with $\mask$ $\mask(\{\goals, \decoys\}\}) = \{\goals\}$.
  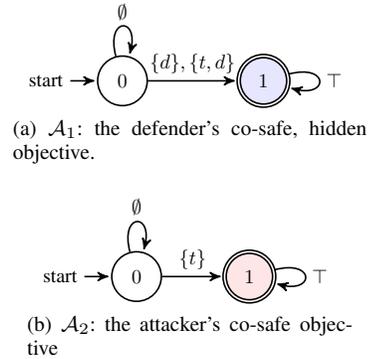
\begin{figure}[h]
  \centering
 \subfloat[$\calA_{1}$: the defender's co-safe, hidden objective.]{
              \begin{tikzpicture}[->,>=stealth',shorten >=1pt,auto,node distance=2.5cm,
                            semithick, scale=0.75, transform shape]
        %\tikzstyle{every state}=[fill=red,draw=none,text=white]
        %\tikzstyle{every state}=[fill=black!10!white]
        \tikzstyle{every state}=[fill=white]
        \node[initial,state]   (0)                      {$0$};
        \node[state,accepting, fill=blue!10]           (1) [right of=0]   {$1$};
        \path[->]   (0) edge              node        {$\{d\},\{t,d\}$}       (1)
        (0) edge [loop above] node        {$\emptyset$}       (0)
                        (1) edge [loop right] node {$\top$} (1);    \end{tikzpicture}
         \label{fig:p1task_preferred}}\qquad\\
    \subfloat[$\calA_2$: the attacker's co-safe objective]{
     \begin{tikzpicture}[->,>=stealth',shorten >=1pt,auto,node distance=2cm,
                            semithick, scale=0.75, transform shape]
        %\tikzstyle{every state}=[fill=red,draw=none,text=white]
        %\tikzstyle{every state}=[fill=black!10!white]
        \tikzstyle{every state}=[fill=white]
        \node[initial,state]   (0)                      {$0$};
        \node[state,accepting, fill=red!10]           (1) [right of=0]   {$1$};
        \path[->]   (0) edge              node        {$\{t\}$}       (1)
        (1) edge [loop right] node {$\top$} (1)
                    (0) edge [loop above] node        {$\emptyset$}       (0);
    \end{tikzpicture}
         \label{fig:p2task}}
        % \qquad \subfloat[$\calA_{1,c}$: the defender's safety objective.]{
        %   \input{reach_goal_p1.tex} 
        %  \label{fig:p1task}}
     \caption{Examples of \ac{dfa}s.}
     \label{fig:ex-specs}
   \end{figure}
   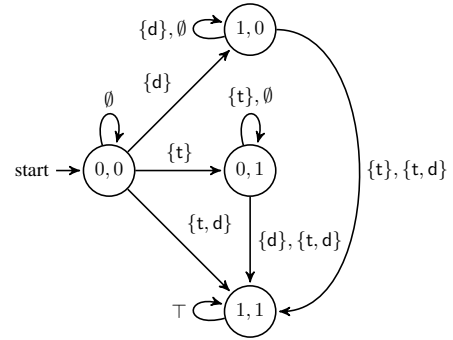
\begin{figure}[h]
     \centering
     %%% Local Variables:
%%% mode: latex
%%% TeX-master: t
%%% End:

    \begin{tikzpicture}[->,>=stealth',shorten >=1pt,auto,node distance=2.5cm,
                            semithick, scale=0.75, transform shape]
        %\tikzstyle{every state}=[fill=red,draw=none,text=white]
        %\tikzstyle{every state}=[fill=black!10!white]
        \tikzstyle{every state}=[fill=white]
        \node[initial,state]   (00)                      {$0,0$};
        \node[state]           (01) [  right of=00]   {$0,1$};
        \node[state]           (10) [above   of=01]   {$1,0$};
        \node[state]           (11) [ below   of=01]   {$1,1$};
        \path[->]   (00) edge              node        {$\{\goals\}$}       (01)
        (00) edge              node        {$\{\decoys\}$}       (10)
                (00) edge   [loop above]           node        {$\emptyset$}       (00)

        (00) edge              node        {$\{\goals,\decoys\}$}       (11)
        (01) edge              node        {$\{\decoys\},\{\goals,\decoys\}$}       (11)
        (10) edge  [bend left = 90]           node        {$\{\goals\}, \{\goals,\decoys\}$}       (11)
        (01) edge [loop above] node {$\{\goals\},\emptyset$} (01)
        (10) edge [loop left] node {$\{\decoys\},\emptyset$} (10)
        (11) edge [loop left] node {$\top$} (11);    \end{tikzpicture}
     \caption{ \label{fig:sync_product} Example of the product automaton $\calA$.}
     \end{figure}
\end{example}

\subsection{Reasoning in hypergames on graphs}
\label{sec:hypergame}
To synthesize deceptive strategy for P1, we construct the following
transition system, called \emph{hypergame transition system}, for P1
to keep track of the history of interaction, the partial satisfaction
of P1's safe and co-safe objectives given the history, as well as what
P1 knows about P2's  perceived partial satisfaction of her co-safe
  objective. % misperception about the labeling. Most importantly, the state in a
% \ac{dfa} for a \ac{ltl} formula are in fact finite-memory states to
% represent the partial satisfaction for the said formula
% \cite{esparza2014ltl}.

\begin{definition}
  \label{def:hypergameTS}
  Given the product automaton
  $ \calA =\langle Q, \Sigma, \delta, I, (\calF_{1},\cosafe),
  (\calF_2,\safe)\rangle $, \ac{dfa}
  $\calA_2 = (Q_2,\Sigma_2, \delta_2, I_2, (F_2,\cosafe))$ for
  P2's co-safe objective, and the game arena
  $G=\langle S,A, T , \calAP, L \rangle$, let
  $L_1: S\rightarrow \Sigma$ be the labeling function of P1 and
  $L_2: S\rightarrow \Sigma$ be the labeling function perceived by
  P2. A hypergame transition system is a tuple
  \begin{multline*}
    \hyperts= \langle (S\times Q\times Q_2), A_1\cup A_2, \Delta, v_0,
   \acc_{1,\cosafe},\acc_{1,\safe}, \acc_{2} \rangle
\end{multline*}
where
\begin{itemize}
\item $V= S\times Q \times Q_2$ is a set of states. A state
  $v= (s,q,q_2)$ includes the state $s$ of the game arena and a state
  $q$ from the product automaton $\calA$ and a state $q_2$ from the
  automaton $\calA_2$. The set of states $V$ is partitioned into $V_1 = S_1\times Q \times Q_2$ and $V_2= S_2\times Q \times Q_2$.
\item $A_1\cup A_2$ is a set of actions.

\item
  $\Delta: V \times (A_1\cup A_2)\rightarrow V$ is the transition function. For a given state $v = (s,q,q_2)$, 
  \[
    \Delta((s,q,q_2),a)=  (s', q',q_2'),
  \]
  where $s' = T(s, a)$, $ q' = \delta(q, L_1(s'))$ is the transition in the automaton $\calA$, $ q_2' = \delta_2(q_2, L_2(s'))$ is a transition in P2's \ac{dfa} $\calA_2$. That is, after reaching the new state $s'$, both P1 and P2 update their \ac{dfa} states to keep track of progress with respect to their specifications.
\item $v_0=(s_0, \delta(I, L_1(s_0)),\delta_2( I_2, L_2(s_0)))$ is the initial state.
  \item $\acc_{1,\cosafe}=(S \times  \calF_{1} \times Q_2)$ is a set of states such that if any state in this set is reached, then the defender achieves his hidden, co-safe objective.   
  \item $\acc_{1,\safe}=( S\times (Q\setminus \calF_2) \times Q_2)$ is a
    set of states such that if the game state is always within $\acc_{1,\safe}$ then the defender achieves his safety objective.
    % in which if the game state always stays, then the
    % defender achieves his safety objective.
\item $\acc_2=(S\times Q\times F_2)$ is the set of states such that if any state in this set is reached, then the attacker achieves her co-safe objective.
  \end{itemize}
\end{definition}

To understand this hypergame transition system. Let's consider a finite sequence of states in the game arena:
\[
  \rho  = s_0s_1s_2\ldots s_n.
\]
From P2's perception, the labeling sequence is:
\[
  L_2(\rho) = L_2(s_0)L_2(s_1)L_2(s_2)\ldots L_2(s_n).
\]
This word $L_2(\rho)$ is evaluated against the formula $\varphi_2$ using the semantics of \ac{ltl} and then  state $q_2 = \delta_2(I_2, L_2(\rho))$ is reached. The perceived progress of P2 is tracked by P1.

From P1's perspective, the labeling sequence is
\[
  L_1(\rho) = L_1(s_0)L_1(s_1)L_1(s_2)\ldots L_1(s_n).\] P1 evaluates
the word $L_1(\rho)$ against two formulas: $\neg \varphi_2$ and
$\hidspec$ and compute $(q_1,q_2^*) = \delta((I_1,I_2),
L_1(\rho))$. The state $q_2^*$ can be different from $q_2$ because that P2's misperception introduces differences in the labeling functions
$L_1\ne L_2$.

% \begin{remark}
%   To ensure that the game terminates when either the security property
%   is violated or the attacker enters a decoy, we can enforce this
%   termination condition by making all the states in
%   $\acc_{1,h}\cup (V\setminus \acc_{1,c})$ sink and define the
%   transition function accordingly. To ensure that the game is
%   terminated only when the security property is violated, we can
%   enforce this termination condition by making all the states in
%   $V\setminus \acc_{1,c}$ sink.
% \end{remark}

% \begin{lemma}
%   Given a path $\rho =v_0v_1\ldots \in V^\omega$ with
%   $v_i\in \acc_{1,p}$ for some $i\ge 0$, if all states in
%   $ V\setminus \acc_{1,c}$ are sink states, then path $\rho$ satisfies
%   both P1's safe and hidden, cosafe objective.
% \end{lemma}
% \begin{proof}
%   The proof is by contradiction. Suppose $\rho$ violates P1's safety
%   objective, then there must exist $k\ge 0$, a sink state
%   $v_k \in V\setminus \acc_{1,c}$ is reached. By definition of sink
%   states, for any $j\ge k$, $v_j= v_k$. Now, we have two cases:

%   case 1: $i <  k$. The defender achieves his cosafe, hidden objective before the safety property is violated.

%   case 2: $i> k$. The defender achieves his cosafe hidden objective after the safety property is violated. The second case is not possible unless there exists a state $v \in \acc_{1,h} \cap (V\setminus \acc_{1,c})$. However, the mere existence of such a state means that P1's objective $\commspec\land \hidspec$ is  
%   \end{proof}
 
Note that due to misperception,   players can be both winning for a given outcome in their respective perceptual games. Consider Example~\ref{ex:simple-specs}, if
there exists a state $s\in S$ labeled to be $L_1(s) = \emptyset$ and
$L_2(s) = \{\goals\}$--that is, the attacker misperceives a non-critical
host as her target, then a path $\rho \in S^\omega$, with
$L_1(s_i)=\emptyset$ for all $i\ge 0$, and $L_2(s_k)=\{\goals\}$ for
some $k\ge 0$ can be considered  to satisfy both the defender's and
attacker's objectives.

  \begin{example}
  \label{ex:simple_game}
  We use a simple game arena with five states, shown in
  Fig.~\ref{fig:arena_ex} to illustrate the construction of hypergame
  transition system. In this game arena, at a square state, P2 selects
  an action in the set $\{b_1,b_2,b_3\}$; at a circle state, P1 selects an
  action in the set $\{a_1,a_2\}$. The labeling functions for two
  players are given as follows:
    \[
    L_1(0)=L_1(1)=L_1(2)=\emptyset; L_1(3)=\{\goals\}; L_1(4) = \{\decoys\};
    \]
    \[
      L_2(0)=L_2(1)=L_2(2)=\emptyset; \quad L_2(3)=L_2(4)=\{\goals\}.
    \]
    % In this example, state 3 is labeled $\goals$ by both P1 and P2
    % because if the state of the network system is in 3, then the
    % attacker has exploited a critical host. However, 
    Note that state 4 is labeled $\goals$ by P2 but $\decoys$ by P1. With this mismatch in the labeling functions, if  state 4 is reached, then the attacker exploits a
    high-fidelity honeypot, known to P1, but falsely believes that
    she has exploited a target host.
    
      \begin{figure}[ht]
           \centering
      %\begin{subfigure}[b]{0.3\textwidth}
           \subfloat[An example of game arena.]{
                 \begin{tikzpicture}[->,>=stealth',shorten >=1pt,auto,node distance=2cm,
                            semithick, scale=0.65, transform shape,  square/.style={regular polygon,regular polygon sides=4}]
        %\tikzstyle{every state}=[fill=red,draw=none,text=white]
        %\tikzstyle{every state}=[fill=black!10!white]
        \tikzstyle{every state}=[fill=white]
        \node[initial,state]   (0)                      {$0$};
        \node[square,draw]           (1) [above right  of=0]   {$1$};
                \node[square,draw]           (2) [below right  of=0]   {$2$};
                \node[state] (3) [right of =1] {$3$};
                \node[state] (4) [right of =2] {$4$};
        \path[->]   (0) edge              node    {$a_1$}           (1)
        (0) edge   node         {$a_2$}  (2)
        (1) edge [bend left] node {$b_3$} (0)
                    (1) edge node {$b_1$} (3)
                    (1) edge node {$b_2$} (4)
                    (2) edge node {$b_1$} (4)
                    (3) edge[loop right] node {} (3)
                    (4) edge[loop right] node {} (4);
    \end{tikzpicture} \label{fig:arena_ex}} \qquad
           \subfloat[A hypergame transition system.]{
                 \begin{tikzpicture}[->,>=stealth',shorten >=1pt,auto,node distance=2cm,
                            semithick, scale=0.65, transform shape,  square/.style={regular polygon,regular polygon sides=4}]
        %\tikzstyle{every state}=[fill=red,draw=none,text=white]
        %\tikzstyle{every state}=[fill=black!10!white]
        \tikzstyle{every state}=[fill=white]
        \node[initial,state, fill=green!10]   (0)                      {$v_0$};
        \node[square,draw, fill=green!10]           (1) [above right  of=0]   {$v_1$};
                \node[square,draw, fill=green!10]           (2) [below right  of=0]   {$v_2$};
                \node[state] (3) [right of =1] {$v_3$};
                \node[state] (4) [right of =2,fill=blue!10] {$v_4$};
        \path[->]   (0) edge              node    {$a_1$}           (1)
        (0) edge   node         {$a_2$}  (2)
                            (1) edge [bend left] node {$b_3$} (0)
                    (1) edge node {$b_1$} (3)
                    (1) edge node {$b_2$} (4)
                    (2) edge node {$b_1$} (4)
                    (3) edge[loop right] node {} (3)
                    (4) edge[loop right] node {} (4);
                    % \node [container,fit=(3) (4), fill=red, opacity=0.05]
                    %       (container) {};
                    %       \node (label) [below of= container] {P2's perceived goal states};
    \end{tikzpicture} \label{fig:hyperts_ex}} \qquad\\
             \subfloat[States in the
           $\hyperts$.]{\adjustbox{raise=2pc}{
         \begin{tabular}{l|l}
         \toprule
         $v_0$ & (0,(0,0),0)\\
             $v_1$ & (1,(0,0),0)\\
                $ v_2$ & (2,(0,0),0)\\
                    $ v_3$ & (3,(0,1),1)\\   
                      $v_4$ & (4,(1,0),1)\\
                      \bottomrule
         \end{tabular}}}
\label{fig:ex_hts}
\caption{An example of game arena and the constructed hypergame
  transition system.  The states in the hypergame transition systems
  are renamed in Table (c) for clarity. }
     \end{figure}
     
     The hypergame transition system is shown in
     Fig.~\ref{fig:hyperts_ex} and constructed from game arena in Fig.~\ref{fig:arena_ex},  \ac{dfa}s $\calA$ in
     Fig.~\ref{fig:sync_product}, and $\calA_2$ in
     Fig.~\ref{fig:p2task}.  After pruning unreachable states, the
     hypergame transition system happened to share the same graph
     topology as the game arena.  Here are two examples to illustrate the
     construction: A transition from $v_1\xrightarrow{b_1}v_3$ is
     generated because of the transitions
     $1\xrightarrow{b_1}3$,
     $(0,0)\xrightarrow{L_1(3) =\{\goals\}}(0,1)$ in $\calA$, and
     $0\xrightarrow{L_2(3) = \{\goals \}}1$ in $\calA_2$.  A
     transition $v_2\xrightarrow{b_1}v_4$ is generated because of
     transitions $2\xrightarrow{b_1} 4$,
     $(0,0)\xrightarrow{L_1(4)=\{\decoys\}}(1,0)$ in $\calA$ and
     $0\xrightarrow{L_2(4)= \{\goals\}} 1$ in $\calA_2$.
     
     In the hypergame transition system, $\acc_2=\{v_3, v_4\}$
     --the set of states that if reached, then P2 believes that she
     has reached the target; $\acc_{1,\safe}=\{v_0, v_1,v_2, v_4\}$
     (shaded in blue and green)--the set of safe states that P1 wants
     the game to stay in; $\acc_{1,\cosafe}=\{v_4\}$(shaded in blue)--the
     set of states that P1 preferred to reach as a hidden, co-safe
     objective.
\end{example}

\subsection{P2's perceptual game and  winning strategy}
\label{sec:perceptualgame-P2}

In the level-2 hypergame, P1 will compute
P2's rational strategy, which is her perceived almost-sure winning
strategy. Then, based on the predicted behavior of P2, P1 can solve
his own almost-sure winning strategies for the safety objective
$\neg \varphi_2$ and the more preferred objective
$\commspec\land \hidspec$, respectively.
% Further,
% as the game is one shot, we assume that the attacker does not update
% her misperception during interaction. The inference mechanism in
% repeated interaction will be analyzed in future.

\begin{definition}
  \label{def:perceivedGame2}
  Given the \ac{dfa}
  $\calA_2 = (Q_2,\Sigma, \delta_2, I_2, (F_2,\cosafe))$ for P2, the
  game arena $G=\langle S,A, T , \calAP, L\rangle$ and P2's labeling
  function $L_2: S\rightarrow \Sigma$, the perceptual game of player 2
  is a tuple
  \begin{multline}
    \game_2= \langle (S\times Q_2), A_1\cup A_2, \Delta_2, (s_0,
    q_{2,0}), S\times F_2 \rangle
\end{multline}
where
\begin{itemize}
\item $ S \times Q_2$ is a set of states. % A state $(s,q_2)$ includes
  % the state $s$ of the game arena and a state $q_2$ from the automaton
  % $\calA_2$.
\item $A_1\cup A_2$ is a set of actions.
\item
  $\Delta_2: (S\times Q_2) \times (A_1\cup A_2)\rightarrow (S\times
  Q_2)$ is the transition function. For a given state $ (s,q_2)$,
  \[
    \Delta_2((s,q_2),a)= (s', q_2'),
  \]
  where $s' = T(s, a)$, $ q_2' = \delta_2(q_2, L_2(s'))$ is a transition in P2's \ac{dfa} $\calA_2$.  % Otherwise, if $v$ is a sink state, then for any action $a\in A_1\cup A_2$, $\Delta (v,a)=v$.
\item $(s_0, q_{2,0})$ with $q_{2,0}=\delta_2( I_2, L_2(s_0))$ is the initial state.
\item $S\times F_2$ is the set of states such that if any state in
  this set is reached, then the attacker perceives that she has
  achieved her co-safe objective.
  \end{itemize}
\end{definition}
  A path of the game graph $\game_2$ is an infinite sequence
  $z_0z_1z_k\ldots $ of states in $S\times Q_2$ such that
  $\Delta_2(z_k,a)= z_{k+1}$ for some $a\in A$ for all $k\ge 0$.  We
  denote the set of paths of $\game_2$ by $Path(\game_2)$.

Due to the determinacy (Thm.~\ref{thm:determined}), we can compute the
solution of this zero-sum game and partition the game states into two sets:
\begin{itemize}
\item P2's perceived winning states for P1:
  $\win_1^2 \subseteq S\times Q_2$, and
\item P2's perceived winning states for herself:
  $\win_2^2 = (S\times Q_2) \setminus \win_1^2$.
  \end{itemize}

  The superscript $2$ means that this solution is for P2's perceptual
  game.   The winning region $\win_2^2$ can be computed from the
  attacker computation, described in Alg.~\ref{alg:zielonka} in
  Appendix with input  game $\game_2$ with 
  $X_1 \coloneqq S_1\times Q_2$, $X_2\coloneqq S_2\times Q_2$,
  $A\coloneqq A_1$ and $B \coloneqq A_2$, $T\coloneqq \Delta_2$.

  Given the winning region, there exists more than one winning strategies that P2 can select. We classify the set of winning strategies for P2 into two sets:
  
  \subsubsection{Greedy strategies for P2}

  In a turn-based co-safe game, for a state in player 2's winning
  region $\win_2^2$, there exists a \emph{memoryless, deterministic,
    sure-winning strategy}
  $\pi_i^\surewin: \win_2^2\rightarrow 2^{A_i}$ such that by following
  the strategy, player $2$ can achieve his/her objective with a
  \emph{minimal number of steps under the worst case strategy} used by
  her opponent.
  
  We call this strategy \emph{greedy} for short. This strategy of P2
  is extracted from the solution of $\game_2$ as follows: Using  Alg.~\ref{alg:zielonka}, we obtain a sequence of sets
  $Z_k \subseteq S\times Q_2$, for $k=1,\ldots, N$, and define the \emph{level} sets: $\level_0 = S\times F_2$, and 
  \begin{equation} 
  \label{eq:levelset} \level_k = Z_k\setminus Z_{k-1}, \text{for } k=1,\ldots, N. \end{equation}
   Intuitively, for any state in
  $\level_k$, P2 has a strategy to ensure to visit a state in
  $\level_0$ in a maximal $k$ steps under the worst case strategy of P1.

  For each $k=1,\ldots, N$, for each $z\in \level_k$, let
  \[ \pi_2^\surewin(z) = \{a\in A_2 \mid \Delta_2(z,a)\in Z_{k-1}\}.
  \]
  In words, by following the greedy strategy, P2 ensures that the
  maximal number of steps to reach $S\times F_2$ is strictly
  decreasing.

\subsubsection{A non-greedy opponent with unbounded memory}When P2 is
allowed to use finite-memory, stochastic strategies, there can be
\emph{infinitely many winning strategies} for P2 given her co-safe
\ac{ltl} objective. % To ensure deceptive planning no matter which
% strategy P2 employs, we will consider an approximation of the set of
% P2's winning strategy.

To see why it is the case, we start with classifying
P2's actions into two sets: Given a history
$\rho \in (S\times Q_2)^\ast$ ending in state $(s,q_2)$,
\begin{itemize}
\item an action $a$ is \emph{perceived to be safe} by P2 if
  $\Delta_2((s,q_2), a) \in \win_2^2$. That is, taking action $a$ will ensure P2 to stay within her perceptual winning region.
\item an action $a$ is \emph{perceived to be sure winning} for P2 if
  $a\in \pi_2^\surewin((s,q_2))$--that is, an action chosen by the
  greedy winning strategy.
  \end{itemize}

  \begin{lemma}
    For a finite-memory, randomized strategy of P2
    $\pi_2:\win_2^2 \rightarrow \dist{A_2}$, P2 can win by
    following $\pi_2$ in her perceptual game if the
    strategy $\pi_2$ satisfies: For every state $(s,q_2)\in \win_2^2$,
    let
    $X = \{\rho \in Path(\game_2) \mid \last(\rho)= (s,q_2)\}$ be a set of paths ending in $(s,q_2)$, it holds that
    \begin{enumerate}
    \item for any $a \in \cup_{\rho \in X}\supp( \pi_2(\rho) )$,
      $ \Delta_2((s,q_2), a) \in \win_2^2$. That is, only safe actions are taken.
   \item
      $\cup _{\rho \in X}\supp(\pi_2(\rho)) \cap
      \pi_2^{\surewin}((s,q_2)) \ne \emptyset$. That is, eventually some sure-winning action will be taken with a nonzero probability.
        \end{enumerate}
Further, this strategy is  \emph{almost-sure winning} for P2.
      \end{lemma}
      \begin{proof}
        The first condition must be satisfied for any winning strategy
        for P2 to stay within the her winning region. The second
        condition is to ensure that a state in $S\times F_2$ is
        eventually visited. To see this, let's use induction: By
        following strategy $\pi_2$, for an arbitrary P1's strategy
        $\pi_1$, let $\{X_0, X_1,X_2,\ldots\}$ be the stochastic
        process of states visited at steps $0,1,2, \ldots$ given
        $(\pi_1,\pi_2)$ in the game $\game_2$. By definition of
        $\pi_2$ and the properties of the game's solution, we have
 the following conditions satisfied:        \begin{inparaenum}[1)]
        \item If it is a state of P2, then the probability of reaching   $\level_{k-1}$ in a finite  number of steps from a state in $\level_k$ is strictly positive:
        $P(X_{t+m}\in \level_{k}\mid X_{t-1}\in \level_{k-1}\land
        X_{t-1}\in S_2\times Q_2)>0$ for some integer $m\ge 0$;
        \item If it is P1's state at level $k$, then for any action of P1, the next state must be in $\level_{\ell}$ for $\ell \le k-1$:
        $P(X_{t}\in \level_{\ell}\mid X_{t-1}\in \level_{k-1}\land
        X_{t-1}\in S_1\times Q_2)=1$;
        \item And the game state is always in $\win_2^2$, 
        $P(X_{t}\in \win_2^2 \mid X_{t-1}\in \win_2^2)=1$.
        \end{inparaenum}
        Let $E_n$
        be the event that ``The level of state is reduced by one in
        $n$ steps''. Given
        $\sum_{n=1}^\infty P(E_n) \ge \sum_{n=1}^\infty p(1-p)^{n-1}=1$
        where $p $ is the minimal probability of $P(E_n)$, the
        probability of reducing the level by one along the path in
        infinitely number of steps is one. In addition, with
        probability one, the level of $X_t$ is finite for all
        $t\ge 0$. Thus, eventually, the level of $X_t$ reduces to zero
        as $t$ approaches infinity.
      \end{proof}

      In words, the almost-sure winning strategy for P2 only allows a
      (perceived) safe action to be taken with a nonzero probability. It
      also enforces that a ``greedy'' action used by the sure-winning
      strategy must be selected with a nonzero probability
      \emph{eventually}.   Since there can be infinitely many such almost-sure winning
         strategies, we take an approximation of the \emph{set} of
         almost-sure winning strategy as a \emph{memoryless set-based
           strategy} as follows.
      \begin{equation}
         \label{eq:approximate-asw}
        \hat{\pi}^\asurewin_2((s,q_2)) = \{a\mid \Delta_2((s,q_2),a )\in \win_2^2 \}.
      \end{equation}
      That is, we assume that at any state $(s,q_2)$, P2 can select
      any action from a set of safe actions to staying within her
      perceived winning region. Note that the progressing action
      $\pi_2^\surewin((s,q_2))$ is also safe, that is,    $\pi_2^\surewin((s,q_2)) \in   \hat{\pi}^\asurewin_2((s,q_2)) $.

      \begin{example}
        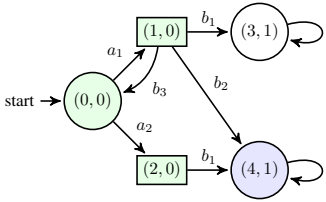
\begin{figure}[ht]
        \centering
              \begin{tikzpicture}[->,>=stealth',shorten >=1pt,auto,node distance=2cm,
                            semithick, scale=0.65, transform shape,  square/.style={rectangle}]
        %\tikzstyle{every state}=[fill=red,draw=none,text=white]
        %\tikzstyle{every state}=[fill=black!10!white]
        \tikzstyle{every state}=[fill=white]
        \node[initial,state, fill=green!10]   (0)                      {$(0,0)$};
        \node[square,draw, fill=green!10]           (1) [above right  of=0]   {$(1,0)$};
                \node[square,draw, fill=green!10]           (2) [below right  of=0]   {$(2,0)$};
                \node[state] (3) [right of =1] {$(3,1)$};
                \node[state] (4) [right of =2,fill=blue!10] {$(4,1)$};
        \path[->]   (0) edge              node    {$a_1$}           (1)
        (0) edge   node         {$a_2$}  (2)
                            (1) edge [bend left] node {$b_3$} (0)
                    (1) edge node {$b_1$} (3)
                    (1) edge node {$b_2$} (4)
                    (2) edge node {$b_1$} (4)
                    (3) edge[loop right] node {} (3)
                    (4) edge[loop right] node {} (4);
                    % \node [container,fit=(3) (4), fill=red, opacity=0.05]
                    %       (container) {};
                    %       \node (label) [below of= container] {P2's perceived goal states};
    \end{tikzpicture}
%%% Local Variables:
%%% mode: latex
%%% TeX-master: t
%%% End:
          \caption{The perceptual game of
            P2.} \label{fig:p2game_ex}
        \end{figure}
        We construct P2's perceptual game graph in
        Fig.~\ref{fig:p2game_ex} using \ac{dfa} $\calA_2$ in
        Fig.~\ref{fig:p2task} and the game arena in
        Fig.~\ref{fig:arena_ex}. Using Alg.~\ref{alg:zielonka}, we
        obtain $Z_0 =\{(3,1), (4,1)\}$,
        $Z_1=\{ (1,0), (2,0)\} \cup Z_0$, and
        $Z_2 =\{(0,0)\} \cup Z_1$. Thus, $\level_0 = Z_0$,
        $\level_1 = \{ (1,0), (2,0)\}$, and $\level_2=\{0,0\}$. The
        greedy winning strategy $\pi_2^\surewin ((1,0))=\{b_1,b_2\}$
        and $\pi_2^\surewin ((2,0))=\{b_1\}$.

        One almost-sure winning strategy for P2 is that
        $\pi_2((1,0),b_i)=\epsilon_i$, for any
        $b_i\in \{b_1,b_2,b_3\}$ and $\pi_2((2,0),b_1)=1$. The
        parameters $\epsilon_i, i=1,2,3$ can be picked arbitrary under the
        constraints $ \epsilon_2>0$ or $\epsilon_3>0$, and
        $\sum_{i=1}^3 \epsilon_i=1$. This strategy of P2 ensures that even if 
        the loop $(0,0)\xrightarrow{a_1} (1,0) \xrightarrow{b_3}(0,0)$
        occurs, it can only occur finitely many times. Eventually,
        $b_1$ or $b_2$ will be selected by P2 to reach $S\times
        F_2$. The approximation of the almost-sure winning strategies
        is $\hat \pi_2^\asurewin((1,0))=\{b_1,b_2,b_3\}$ and
        $\hat \pi_2^\asurewin((2,0))=\{b_1\}$.
      \end{example}

      \subsection{Synthesizing P1's deceptive winning strategy}
  \label{subsec:deceptive_win}
Next, we use the hypergame transition system, a P2's strategy, to
compute a subgame for P1.

\begin{definition}[$\pi_i$-induced subgame graph]
  Given the graph of a game
  $\game = \langle V_1\cup V_2, A_1\cup A_2, \Delta, v_0 \rangle$
  and a strategy of player $i$, $\pi_i: V_i\rightarrow 2^{A_i}$, a $\pi_i$-induced
  subgame graph, denoted $\game_{\slash \pi_i} $, is the game graph
  $ \langle V_1\times V_2, A_1\cup A_2, \Delta_{\slash \pi_i}, v_0 \rangle$ where
  \begin{align*}
    \Delta_{\slash \pi_i}(v,a)& =\Delta(v,a); \quad \forall a\in A_k, k\ne i; \\
    \Delta_{\slash \pi_i}(v,a) &=\Delta(v,a); \quad \forall a\in \pi_i(v); \\
    \Delta_{\slash \pi_i}(v,a) & \uparrow; \quad \forall a\in  A_i\setminus \pi_i(v).
  \end{align*} where $\uparrow$ means that the function is undefined for the given input.
\end{definition}
The subgame graph restricts player $i$'s actions to these
allowed by strategy $\pi_i$. It does not restrict player $k$'s
actions. Given a subgame graph $\game_{\slash \pi_i}$ and
another player $k$'s strategy $\pi_k: V_k\rightarrow 2^{A_k}$, we can
compute another subgame graph induced by $\pi_k$ from
$\game_{\slash \pi_i}$ and denote the subgame as
$\game_{\slash \pi_i, \pi_k}$.

Now, assuming that P2 follows the perceived winning strategy
$\pi_2: S\times Q_2\rightarrow 2^{A_2}$, we can construct an induced
subgame graph from the hypergame transition system $\hyperts$ using
P2's strategy defined by: $\pi'_2(v) \coloneqq \pi_2(s,q_2)$ for each
$v=(s,q_1,q_2)\in V_2$. Slightly abusing the notation, we still use $\pi_2$ to refer to P2's strategy defined over domain $V_2$.
% \begin{definition}
%   Given the hypergame transition system
%   $ \hyperts= \langle V_1\cup V_2, A_1\cup A_2, \Delta,
%   v_0, \acc_{1,p}, \acc_{1,c}, \acc_2 \rangle $ and a memoryless,
%   set-based strategy of P2 $\pi_2: S\times Q_2 \rightarrow 2^{A_2}$, a
%   $\pi_2$-induced subgame is a tuple,
%   \[ \hyperts_{\slash \pi_2} = \langle (S\times Q_1\times Q_2),
%     A_1\cup A_2, \Delta_{\slash \pi_2}, v_0, \acc_{1,p}, \acc_{1,c}
%     \rangle\] where the transition function is revised as follows: For
%   each $v= (s,q,q_2)$,
% \begin{align*}
%   \Delta_{\slash \pi_2}(v,a)& =\Delta(v,a); \quad \forall a\in A_1; \\
%   \Delta_{\slash \pi_2}(v,a) &=\Delta(v,a); \quad \forall a\in \pi_2(s,q_2); \\
%   \Delta_{\slash \pi_2}(v,a) & \uparrow; \quad \forall a\in  A_2\setminus \pi_2(s,q_2).
% \end{align*}
% where $\uparrow$ means the function is undefined.
% \end{definition}

By replacing $\pi_2$ to be either 1) the  greedy policy
$\pi_2^\surewin$ or 2) the approximation of almost-sure winning
strategies $\hat \pi_2^\asurewin$, we can solve P1's winning strategy with
respect to different objectives, $\commspec$ or
$\commspec\land \hidspec$, leveraging the information about P2's misperception.

% 1) the safe
% objective $\varphi_1$ where P1 aims to ensure staying in $\acc_{1,c}$;
% or 2) the conjunction of safe and co-safe objective
% $\varphi_{1,p}\land \commspec$ where P1 aims to ensure staying in
% $\acc_{1,c}$ while reaching a subset of states $\acc_{1,h}\cap \acc_{1,c}$.

We present a two-step procedure to solve P1's deceptive sure-winning strategy.
\begin{enumerate}
\item Step 1: Solve the $\pi_2$-induced subgame for P1 with respect to
  the safety objective:
  \[ \hyperts_{\slash \pi_2} = \langle V,
    A_1\cup A_2, \Delta_{\slash \pi_2}, v_0, ( \acc_{1,\safe},\safe)\]
  using Alg.~ \ref{alg:safe} with input
  $X_1 = V_1$, $X_2 =V_2$ and
  $T = \Delta_{\slash \pi_2}$, and let $B = \acc_{1,\safe}$. The
  outcome is a tuple $(\win_{1,\safe}, \pi_{1,\safe})$--that is, a set
  of states from which P1 ensures that the safety objective $\neg \varphi_2$ can be
  satisfied, by following the winning, set-based strategy
  $\pi_{1,\safe}:\win_{1,\safe}\rightarrow 2^{A_1}$.
\item Step 2: Compute the $\pi_{1,\safe}$-induced sub-game from
  $\hyperts_{\slash \pi_2}$, denoted as
  \begin{multline*}
    \hyperts_{\slash (\pi_{1,\safe},\pi_2)} \\= \langle V\cap \win_{1,\safe},
    A_1\cup A_2, \Delta_{\slash (\pi_{1,\safe}, \pi_2)},\\
    v_0, ( \acc_{1,\cosafe} \cap \win_{1,\safe},\cosafe)
    \rangle.\end{multline*} Then, we  solve the subgame for P1's co-safe
  objective using Alg.~\ref{alg:zielonka} with input
  $X_1= V_1 \cap \win_{1,\safe},$,
  $X_2 = V_2\cap \win_{1,\safe},$ and
  $T = \Delta_{\slash(\pi_{1,\safe}, \pi_2)}$, and let
  $F = \acc_{1,\cosafe}  \cap \win_{1,\safe}$. The outcome is a tuple
  $(\win_{1,\cosafe}, \pi_{1,\cosafe})$--that is, a set of states from
  which P1 ensures that both the safety and co-safe, hidden objectives
  can be satisfied, by following the winning, set-based strategy
  $\pi_{1,\cosafe}:\win_{1,\cosafe}\rightarrow 2^{A_1}$.  Note that
  safety objective is satisfied because P1 only can select actions
  allowed by his safe strategy $\pi_{1,\safe}$.
\end{enumerate}

The next Lemma shows that for any safe or co-safe objective of P1, if a
strategy is winning for P1 against the approximation of almost-sure
winning strategies of P2 (see \eqref{eq:approximate-asw} for the definition), then the same strategy is winning for P1
against any almost-sure winning strategy of P2.

  \begin{lemma}
    Consider a game
    $\game = (V_1\cup V_2, A_1\cup A_2, \Delta, v_0)$
    and two strategies of Player 2,
    $\pi_2:V^\ast \rightarrow \dist{A_2}$ is finite-memory and
    randomized; and $\hat{\pi}_2: V\rightarrow 2^{A_2}$ is a
    stationary, set-based, and deterministic. If these two strategies
    satisfies that for any $\rho= v_0v_1\ldots v_k \in V^\ast $,
    $\pi_2(\rho,a)>0$ if and only if $a\in \hat{\pi}_2 (v_k)$. That
    is, $\supp(\pi_2(\rho)) = \hat{\pi}_2 (v_k)$.  Then, given an initial state $v_0\in V$ and any subset $F\subseteq V$, if P1 has a
    winning strategy $\pi_1$ in the $\hat{\pi}_2$-induced game for
    objective $(F,\safe)$ (or $(F, \cosafe)$), then P1 wins even if P2 follows strategy $\pi_2$.
    \end{lemma}
    \begin{proof}
      We consider two cases:

      Case 1: P1's objective is given by $(F, \safe)$--that is, P1 is
      to ensure the game states to stay in the set $F$. By definition, the winning strategy of P1 ensures that the game stays within a set $\win_1\subseteq F$ of safe states, no matter which action P2 selects using $\hat \pi_2$. For any
      history $\rho=v_0v_1\ldots v_k \in V^\ast$, for any action $a $
      that $\pi_2(\rho)$ will select with a nonzero probability, then
      the resulting state is still within $\win_1\subseteq F$, because
      $a \in \hat{\pi}_2(v_k)$. Thus, P1 ensures to satisfy the safety
      objective with strategy $\pi_1$ even against P2's strategy
      $\pi_2$.

      Case 2: P1's objective is given by $(F, \cosafe)$--that is, P1
      is to reach set $F$. In this case, the winning region of P1 is
      partitioned into level sets (see
      Alg.~\ref{alg:zielonka} and the definition of level sets \eqref{eq:levelset}). Given a history $\rho\in V^\ast$, if
      $v_k $ is P1's turn and $v_k\in \level_i$, then by construction,
      there exists an action of P1 to reach $\level_{i-1}$. Otherwise,
      $v_k$ is P2's turn and $v_k\in \level_i$, by construction, for
      any action of P2  in $\hat \pi_2(v_k)$, the next state
      is in $\level_{i-1}$. While P2 follows $\pi_2$, the probability
      of reaching $\level_{i-1}$ in one step is
      \begin{align*}
        &\sum_{a\in A_2} \pi_2(\rho, a) \mathbf{1}(\Delta(v_k,a)\in
          \level_{i-1})\\
        = &\sum_{a\in {\supp(\pi_2(\rho))} } \pi_2(\rho, a)
            \mathbf{1}(\Delta(v_k,a)\in \level_{i-1})\\
        = & \sum_{a\in {\supp(\pi_2(\rho))} } \pi_2(\rho, a)=1, \end{align*}
        where the first equality is because only actions in the support of $\pi_2(\rho)$ can be chose with nonzero probabilities; and the second equality is 
        because $\supp(\pi_2(\rho)) \subseteq \hat\pi_2(s_k) $
        and         $\mathbf{1}(\Delta(v_k,a)\in \level_{i-1}) =1$ for any
        $a\in \hat \pi_2(v_k)$. Thus, the level will be strictly decreasing every time an action is taken by P1 or P2 and the level of any state in the winning region $\win_1$ is finite. When the level reaches zero, P1 visits $F$.
      \end{proof}
    
  % \begin{theorem}
  %   For any finite-memory, randomized, permissive winning strategy
  %   $\pi^\asurewin_2$ of P2 in her perceptual game, if P1 has a winning strategy
  %   $\pi_{1,\safe}$ for safety (resp. $\pi_{1,h}$ for the conjunction of
  %   safe and cosafe objectives) defined for a history
  %   $\rho = s_0s_1\ldots s_k$ in the subgame
  %   $\hyperts_{\slash \hat{\pi}^\asurewin_2}$ (resp.
  %   $\hyperts_{\slash (\pi_{1,\safe}, \hat{\pi}^\asurewin_2)}$), then P1
  %   can ensure to satisfy $\commspec$ (resp.
  %   $\commspec\land \hidspec$) in the game given P2's
  %   strategy $\pi_2^\asurewin$.
  % \end{theorem}
  % \begin{proof}
  %   The winning condition of P1 is that no matter which action P2
  %   selects given a P2's state, the resulting state will be in the
  %   winning region of P1. For any history $\rho = s_0s_1\ldots s_k$,
  %   an action that P2 plays according to $\pi_2^\asurewin$ is within the set
  %   $ \hat{\pi}^\asurewin_2$ by definition of
  %   $ \hat{\pi}^\asurewin_2$. Thus, if for any action
  %   $ a\in \hat{\pi}^\asurewin_2(s_k, q_{2})$ with
  %   $q_2= \delta_2(I_2, L_2(\rho))$, the next state is in P1's winning
  %   region, then for any action $a$ with $\pi_2^\asurewin(\rho,a)>0$, because
  %   $a\in \hat{\pi}^\asurewin_2(s_k, q_2)$, the next state is in
  %   P1's winning region.
  %   \end{proof}

  % By solving the sure-winning strategy for the defender in this induced hypergame transition system, we take into account of all possible rational strategies that the attacker can use given her misperception of the labeling function.

\begin{remark}
  In our analysis, the defender is playing against all possible 
  strategies that can be used by the attacker. The deceptive winning
  strategy computed for the defender can be conservative for any fixed
  attack strategy used by the attacker. For example, if the attack
  takes a set of almost-sure winning actions uniformly at random (all
  actions are equally winning as she perceives), then it may leave the
  opportunity for the defender to ensure, with a positive probability,
  the attacker will be lured into a honeypot, and at the same time,
  the defender can ensure, with probability one, the safety objective
  of the system is satisfied. This chance-winning strategy requires
  further analysis of positive winning in games and could be explored
  in the future.
\end{remark}

Next, we use the toy example to demonstrate the computation of player 1's almost-sure winning strategy with deception.

\begin{example}
  Let's revise the simple example in Ex.~\ref{ex:simple_game} to
  illustrate the computation of deceptive winning regions for P1. We added one transition
  $2\xrightarrow{b_2} 1$. The resulting hypergame transition system is shown in
  Fig.~\ref{fig:revised_ex}. Note that the perceptual game of P2 is
  changed in similar way by adding a transition from
  $(1,0)\xrightarrow{b_2}(1,0)$. We omit the
  figure here.

  \begin{figure}[h]
    \subfloat[The example of $\hyperts$ after revising a transition.]{
        \begin{tikzpicture}[->,>=stealth',shorten >=1pt,auto,node distance=2cm,
                            semithick, scale=0.65, transform shape,  square/.style={regular polygon, regular polygon sides=4}]
        %\tikzstyle{every state}=[fill=red,draw=none,text=white]
        %\tikzstyle{every state}=[fill=black!10!white]
        \tikzstyle{every state}=[fill=white]
        \node[initial,state, fill=green!10]   (0)                      {$v_0$};
        \node[square,draw, fill=green!10]           (1) [above right  of=0]   {$v_1$};
                \node[square,draw, fill=green!10]           (2) [below right  of=0]   {$v_2$};
                \node[state] (3) [right of =1] {$v_3$};
                \node[state] (4) [right of =2,fill=blue!10] {$v_4$};
        \path[->]   (0) edge              node    {$a_1$}           (1)
        (0) edge   node         {$a_2$}  (2)
                            (1) edge [bend left] node {$b_3$} (0)
                    (1) edge node {$b_1$} (3)
                    (1) edge node  {$b_2$} (4)
                    (2) edge node {$b_1$} (4)
                    (2) edge [draw,dashed, bend right] node [right] {$b_2$} (1)
                     (3) edge[loop right] node {} (3)
                    (4) edge[loop right] node {} (4);
                    % \node [container,fit=(3) (4), fill=red, opacity=0.05]
                    %       (container) {};
                    %       \node (label) [below of= container] {P2's perceived goal states};
    \end{tikzpicture}
%%% Local Variables:
%%% mode: latex
%%% TeX-master: t
%%% End:
    \label{fig:revised_ex}}
   \subfloat[The $\pi_2^\surewin$-induced $\hyperts$.]{
      \begin{tikzpicture}[->,>=stealth',shorten >=1pt,auto,node distance=2cm,
                            semithick, scale=0.65, transform shape,  square/.style={regular polygon,regular polygon sides=4}]
        %\tikzstyle{every state}=[fill=red,draw=none,text=white]
        %\tikzstyle{every state}=[fill=black!10!white]
        \tikzstyle{every state}=[fill=white]
        \node[initial,state, fill=green!10]   (0)                      {$v_0$};
        \node[square,draw, fill=green!10]           (1) [above right  of=0]   {$v_1$};
                \node[square,draw, fill=green!10]           (2) [below right  of=0]   {$v_2$};
                \node[state] (3) [right of =1] {$v_3$};
                \node[state] (4) [right of =2,fill=blue!10] {$v_4$};
        \path[->]   (0) edge              node    {$a_1$}           (1)
        (0) edge   node         {$a_2$}  (2)
        (1) edge node {$b_1$} (3)
                            (1) edge node {$b_2$} (4)
                    (2) edge node {$b_1$} (4)
                    (3) edge[loop right] node {} (3)
                    (4) edge[loop right] node {} (4);
                    % \node [container,fit=(3) (4), fill=red, opacity=0.05]
                    %       (container) {};
                    %       \node (label) [below of= container] {P2's perceived goal states};
    \end{tikzpicture}
%%% Local Variables:
%%% mode: latex
%%% TeX-master: t
%%% End:
    \label{fig:revised_ex_induced}}
    \end{figure}

    With this change, it can be verified that the winning region of P2 in her peceptual game is not changed and the sure-winning strategy is still $\pi_2^\surewin(1,0)=\{b_1,b_2\}$, $\pi_2^\surewin(2,0)=\{b_1\}$. The approximation of almost-sure winning strategy is now $\hat \pi^\asurewin_2(1,0)=\{b_1,b_2,b_3\}$ and $\hat \pi_2^\asurewin(2,0)=\{b_1,b_2\}$. 

    Case 1: P2 is greedy. The $\pi_2^\surewin$-induced subgame is
    shown in Fig.~\ref{fig:revised_ex_induced}. By applying
    Algorithm~\ref{alg:zielonka}, we obtain P1's deceptive winning
    region and winning strategy $\win_{1,\safe}= \{v_0,v_2,v_4\}$ and
    $\pi_1(v_0)=a_2$. This strategy also ensures that P1 can satisfy
    the hidden, co-safe objective and lead P2 to visit the decoy.
    
    Case 2: P2 is randomized. The $\hat \pi_2^\asurewin$-induced subgame is exactly the one in Fig.~\ref{fig:revised_ex}. By applying
    Algorithm~\ref{alg:zielonka}, we obtain P1's deceptive winning
    region and winning strategy $\win_{1,\safe}= \emptyset$. Thus, P1 cannot achieve either his safety or hidden co-safe objectives when P2 is not greedy and uses a finite-memory, randomized strategy.
         \end{example}

         \subsection{Complexity analysis}
         \label{subsec:complexity}
         Algorithm~\ref{alg:zielonka} and Algorithm~\ref{alg:safe} run in $\mathcal{O}(\abs{X}+\abs{T})$ where $\abs{X}$ is the number of states and  $\abs{T}$ is the number of transitions in the game. Thus, the time complexity of solving P1's deceptive winning region for  safety and the preferred objectives is $\mathcal{O}(\abs{V}+\abs{\Delta})$. To compute a policy $\pi_i$-induced subgame, the time complexity is $\Theta(\abs{\mathsf{Dom}(\pi_i)})$ where $\mathsf{Dom}(\pi_i)$ is the domain of policy $\pi_i$.  
         
         \section{Experimental result}
\label{sec:experiment}
We demonstrate the effectiveness of the proposed synthesis methods in a synthetic network.    The topology of the network is manually generated. The set of vulnerabilities is defined  with the pre- and post-conditions of concrete vulnerability instances  in \cite{Sheyner2004}. The generation of the attacker graph is based on multiple prerequisite attack graph \cite{ingols2006practical} with some modification. We assume that the attacker cannot carry out attacks from multiple hosts at the same time. We perform two experiments on two networks (small and large) and different sets of attacker and defender objectives.

\subsection{Experiment 1: A small network with simple attacker/defender objectives}

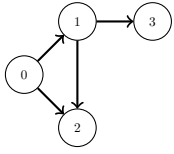
\begin{figure}[ht]
\centering
%%% Local Variables:
%%% mode: latex
%%% TeX-master: t
%%% End:

\begin{tikzpicture}[align=center,  node distance=2cm, scale=0.5, transform shape]
    \node[main node] (0) {$0$};
    \node[main node] (1) [above right of =0]  {$1$};
    \node[main node] (2) [below right of =0] {$2$};
    \node[main node] (3) [right of = 1] {$3$};
    
    \path[->,draw,thick]
    
    (0) edge node {} (1)
    (0) edge node {} (2)
    (1) edge node {} (2)
    (1) edge node {} (3);
\end{tikzpicture}
\caption{\label{fig:host_connectivity_graph}The connectivity graph of hosts in a small network.}
\end{figure}
Formally, the network consists of a list $\hosts$ of hosts, with a connectivity graph shown in Fig.~\ref{fig:host_connectivity_graph}. Each host runs a subset $\servs = \{0,1,2\}$ of services. A user in the network can have one of the three login credentials $\mbox{credentials}= \{0,1,2\}$ standing for ``no access'' (0), ``user'' (1), and ``root'' (2). There are a set of vulnerabilities in the network, each of which is defined by a pre-condition and a post-condition. The pre-condition is a Boolean formula that specifies the set of logical properties to be satisfied for an attacker to exploit the vulnerability instance. The post-condition is a Boolean formula that specifies the logical properties that can be achieved after the attacker has exploited that vulnerability. For example, a vulnerability named ``IIS\_overflow'' requires the attacker to have a credential of user or a root, and the host running IIS webserver. After the attacker exploits this vulnerability, the service will be stopped and the attacker gains the credential as a root user. The set of vulnerabilities are given in Table~\ref{tb:vuldefine} and  are generated based on the vulnerabilities described in \cite{Sheyner2002}. Note that the game-theoretic analysis is not limited to these particular instances of vulnerabilities in \cite{Sheyner2002} but rather used these as a proof of concept.

The defender can temporally suspend noncritical services from servers. To incorporate this defense mechanism, we assign each host a set of noncritical services that can be suspended from the host.  In Table~\ref{tb:networkstatus_simple}, we list the set of services running on each host, and a set of noncritical services that can be suspended by the defender. Other defenses can be considered. For example, if the network topology can be reconfigured online, then the state in the game arena should keep track of the current topology configuration of the network. In our experiment, we consider simple defense actions. The method extends to more concrete defense mechanisms. 

\begin{table}[ht]
\centering
\caption{The pre- and post-conditions of vulnerabilities. }
\label{tb:vuldefine}
\begin{tabular}{c|l }
\toprule
     vul id. &  pre- and post- \\
     \toprule
     0 & Pre : $c\ge 1$, service $0$ running on target host, \\
     &Post : $c=2$, stop service $0$ on the target, reach target host.\\
     \hline
     1 & Pre: $c\ge 1$, service $1$ running on the target host, \\
     & Post : reach target host. \\
     \hline
     2 & Pre: $c\ge 1$, service $2$ running on the target host \\
     & Post : $c=2$,  reach target host.\\
     \bottomrule
\end{tabular}
\end{table}

\begin{table}[ht]
    \centering    \caption{The network status and the defender's options}
    
    \begin{tabular}{c|c|c}
    \toprule
         host id. & services & non-critical services  \\
         0&$\{1\} $& $\emptyset$\\
         1& $\{0,1\}$ &$\{1\}$\\
         2& $\{0,1,2\}$ & $\{1\}$\\
         3 & $\{0,1,2\}$ &$\{0,1,2\}$\\
         \bottomrule
    \end{tabular}
    \label{tb:networkstatus_simple}
\end{table}
Given the network topology, the services and vulnerabilities, a state of the game arena is a tuple
\[
  (h, c, \turn, \nwstate),
\]
where $h\in \hosts$ is the current host of the attacker, $c \in \creds$ is the current
credential of the attacker on that host $h$, $\turn$ is a turn
variable indicating whether the attacker ($\turn= 1$) takes an
action or the defender does ($\turn= 0$). The last component
$\nwstate$ specifies the network condition as a list of
$(h_i, \serv_i)$ pairs that gives a set $ \serv_i$ of
services running on host $h_i$, for each $h_i\in \hosts$. We
denote the set of possible network conditions in the network as $\nwstates$.

The attacker, at a given
attacker's state, can exploit any existing vulnerability on the
current host.  The defender, at a defender's state, can choose to suspend
a noncritical service on any host in the network. Given the defenses and attacks, a sampled path in the arena is illustrated as follows.

  %\begin{minipage}{\linewidth}
\begin{multline*}
    (0, 1, 1, \{0: \{1\}, \mathbf{ 1: \{0, 1\}}, 2: \{0,1, 2\}, 3: \{0, 1, 2\}\})
  \\
  \xrightarrow{\textcolor{red}{(1,0)}}(1, 2, 0, \{0: \{1\}, \mathbf{1: \{1\}}, 2: \{0,1, 2\}, 3: \{0, 1, 2\}\})\\
  \xrightarrow{\textcolor{blue}{(2,1)}}
  (1, 2, 1, \{0: \{1\}, 1: \{1\}, \mathbf{2: \{0, 2\}}, 3: \{0,1, 2\}\}),
\end{multline*}
  % \end{minipage}

This specifies that the attacker is in source host $h=0$ with user access $c= 1$ and has her turn $\turn= 1$ to take an action. She chooses to attack host $1$ with the vulnerability $0$ (action $(1,0)$). The vulnerability $0$'s pre-condition requires  that the attacker has a user/root access  on the source host and service $0$ is running on the target host. After taking the action, the attacker reaches the target host $h'=1$, stopped the service and change her access to root $c' = 2$ (see Table~\ref{tb:vuldefine}). Now it is the defender's turn $\turn' =0$. The service $0$ is stopped in host $1$, resulting the only change in the $\nwstate$ as shown in boldface. Then, the defender takes an action $(2, 1)$, indicating the action to
stop service $1$ on host $2$. This action results in a change in the
network status, as shown in boldface. The arena is generated to
consider all possible actions that can be taken by the attacker or
defender. If at a given state there is no action available, \ie, no
vulnerability to exploit or action to defense, the player selects a
``null'' action, switching the turn to the other player. The full arena is shown in Fig.~\ref{fig:arena_simple} in Appendix. At the beginning of
this game, the attacker is at host 0 with user access. Each host runs
all available services. The state colored in red  in Fig.~\ref{fig:arena_simple} is the initial state in the arena.

In the following experiments, we consider players' objectives in Example~\ref{ex:simple-specs}.
 The following labeling functions are defined: For the defender, we have
\[
L_1(h,c,\turn, \nwstate)=\left\{\begin{array}{ll} 
\{\goals\} &\text{ if } h=3 \text{ and }c\ge 1,\\
\{\decoys\} &\text{ if }h =2 \text{ and } c\ge 1,\\
\emptyset & \text{ otherwise.}\end{array}\right.
\]
In words, the decoy is set to host $2$. The attacker violates the safety objective if she reaches host $3$ with user or root credentials. For the  attacker, we have 
\[
L_2(h,c, \turn, \nwstate)=\left\{\begin{array}{ll} 
\{\goals\} & \text{ if } h=3, 2 \text{ and }c\ge 1,\\
\emptyset & \text{ otherwise.}\end{array}\right.
\]
That is, the attacker is to reach either host 2 or host 3, which are perceived to be critical hosts by the attacker. 

For this experiment, we compare the set of states from which P1 has a winning strategy to achieve his two objectives: safety $\commspec$ and preferred $\commspec\land \hidspec$. Let $\win_1$ (resp. $\win_1^\succ$) be the winning region of P1 given objective $\commspec$ (resp. $\commspec\land \hidspec$). As shown in Table~\ref{tb:compare-simple-specs}, the total number of states in the hypergame transition system is $259$, when P2 has no misperception and plays the zero-sum game against P1's two objectives, there does not exist a winning strategy for P1 given the initial state of the game. When P2 perceives the decoy to be a critical host (see the labeling function $L_2$) and uses the sure-winning greedy strategy, P1's winning regions for both objectives are larger and include the initial state, that is, P1 has a winning strategy to ensure the safety objective and force P2 to visit the decoy host. However, under this misperception, when P2 selects a randomized, finite memory strategy, there does not exist a winning strategy for P1 given the initial state of the game. This means a greedy adversary is easier to deceive than a non-greedy adversary. When P1 uses  the policy $\hat \pi_2^\asurewin$ to approximate P2's randomized strategies, his strategies are conservative for both objectives. 

In Appendix, Figure~\ref{fig:game_win} depicts the winning regions computed using the proposed method and P2's sure-winning greedy policy, with three partitions of the states in the hypergame transition system: \begin{inparaenum}[1)]
\item A set of states colored in blue or orange are the states where P2 perceives herself  to be winning given the co-safe objective and P1 can deceptively win against P2's co-safe objective given P2's greedy strategy.
\item A set of states colored in red are the states where P2 perceives herself to be winning and P1 cannot deceptively win against P2.
\item A set of states colored in yellow are the states where P2 perceives herself to be losing given the co-safe objective and P1 can deceptively win the safety objective. The union of blue and yellow states is $\win_1$.
\item A set of four states colored in orange and grey are those in P1's deceptive winning region $\win_1$ under any sure winning strategy of P2 but not in $\win_1$ under any randomized, finite memory, almost-sure winning strategy for P2.
\item The initial state is marked with an incoming arrow and colored in grey. 
\end{inparaenum}
It is noted that the set of yellow states is isolated from P2's perceived winning region. This is because we restricted P1's actions to only the ones allowed by his winning strategy. By following these strategies, P2 is confined to a set of states within $\win_1$.

\begin{table}\centering
  \caption{
                               \label{tb:compare-simple-specs} Comparison P1's winning regions with/without deception.}
    \resizebox{0.8\textwidth}{!}{\begin{minipage}{\textwidth}                       
                                          \begin{tabular}{c | cc | cc|cc}
\toprule
                                     $\abs{V}$ &      \multicolumn{2}{c}{No Misperception} & \multicolumn{2}{c}{Greedy} &
                                 \multicolumn{2}{c}{Randomized} 
                                 \\
                                              \hline
                                       &   $\win_1 $ & $\win_1^\succ$    &   $\win_1$ &  $\win_1^\succ$  &    
                                     $\win_1$ &  $\win_1^\succ  $   \\
                                               259 &   187, lose & 116, lose  & 191, win &  134, win & 187, lose   & 130, lose \\
                                          \bottomrule
                                          \end{tabular}
      \end{minipage}}
                                        \end{table}

\subsection{Experiment: A larger network with more complex attacker/defender objectives}

In this experiment, we consider a larger network with 7 hosts, with the connectivity graph shown in Fig.~\ref{fig:network_graph_large}.  P2's co-safe objective is given by an \ac{ltl} formula
$\varphi_2\coloneqq \Eventually A \land \Eventually B$--that is, visiting hosts labeled
$A$ and $B$ eventually. The \ac{dfa} $\calA_2$ that represents this
specification is given in Fig.~\ref{fig:dfa2_large}. The defender's
co-safe objective $\calA_1$ is the same as in
Fig.~\ref{fig:p1task_preferred}, which is to force the attacker to
visit a decoy. 
In Table~\ref{tb:networkstatus_large}, we list the set of services running on each host, and a set of noncritical services that can be suspended by the defender.
\begin{table}[ht]
  \centering \caption{The set of available services and non-critical services on each host.}
    \begin{tabular}{c|c|c}
    \toprule
         host id. & services & non-critical services  \\
         0&$\{1,2\} $& $\{2\}$\\
         1& $\{0,1,2\}$ &$\{1\}$\\
         2& $\{0,1,2\}$ & $\{1\}$\\
      3 & $\{0,1,2\}$ &$\{1\}$\\
      4 & $\{0,1,2\}$ &$\{0,1,2\}$\\
      5 & $\{1,2\}$ &$\{2\}$\\
      6 &  $\{0,1,2\}$ &$\{1\}$.\\
      \bottomrule
    \end{tabular}
    \label{tb:networkstatus_large}
\end{table}

\begin{figure}[ht]
\centering
%%% Local Variables:
%%% mode: latex
%%% TeX-master: t
%%% End:

%%% Local Variables:
%%% mode: latex
%%% TeX-master: t
%%% End:

\begin{tikzpicture}[align=center,  node distance=2cm, scale=0.5, transform shape]
    \node[main node] (0) {$0$};
    \node[main node] (1) [above right of =0]  {$1$};
    \node[main node] (3) [below right of =0] {$3$};
    \node[main node] (2) [below right of = 3] {$2$};
    \node[main node] (4) [right of = 3] {$4$};
    \node[main node] (5) [right of =4] {$5$};
    \node[main node] (6) [right of = 1] {$6$};
     \path[->,draw,thick]
     (0) edge node {} (1)
    (0) edge node {} (3)
    (1) edge node {} (6)
    (3) edge node {} (2)
    (3) edge node {} (4)
    (2) edge node {} (5)
    (6) edge node {} (4)
    (4) edge [bend left] node {} (6)
    (4) edge node {} (5)
    (5) edge node {} (6);
\end{tikzpicture}
\caption{\label{fig:network_graph_large}The connectivity graph of hosts in a network of seven hosts.}
\end{figure}
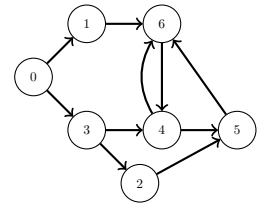 

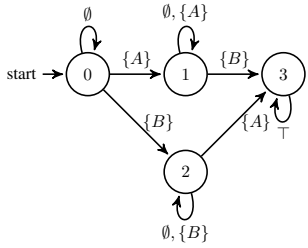
\begin{figure}[ht]
\centering
%%% Local Variables:
%%% mode: latex
%%% TeX-master: t
%%% End:
    \begin{tikzpicture}[->,>=stealth',shorten >=1pt,auto,node distance=2cm,
                            semithick, scale=0.65, transform shape]
        %\tikzstyle{every state}=[fill=red,draw=none,text=white]
        %\tikzstyle{every state}=[fill=black!10!white]
        \tikzstyle{every state}=[fill=white]
        \node[initial,state]   (0)                      {$0$};
        \node[state]           (1) [ right of=0]   {$1$};
        \node[state]           (2) [below  of=1]   {$2$};
        \node [state] (3) [right of =1] {$3$};
        \path[->]
        (0) edge              node        {$\{A\}$}       (1)
        (0) edge [loop above] node        {$\emptyset$}       (0)
        (0) edge              node  [right]      {$\{B\}$}       (2)
        (1) edge [loop above] node        {$\emptyset,\{A\}$}       (1)
        (2) edge [loop below] node        {$\emptyset,\{B\}$}       (2)
        (1) edge  node        {$\{B\}$}       (3)
        (2) edge  node   [right]     {$\{A\}$}       (3)
                (3) edge [loop below]  node        {$\top$}       (3);
        \end{tikzpicture}
\caption{\label{fig:dfa2_large} The \ac{dfa} $\calA_2$ for P2's co-safe objective.}
\end{figure}

The following labeling functions are defined for P1 and P2:
\[
L_1(h,c,\turn, \nwstate)=\left\{\begin{array}{ll} 
\{A \} &\text{ if } h= 2 \text{ and }c\ge 1,\\
\{B\} &\text{ if }h = 6 \text{ and } c\ge 1,\\
  \{\decoys\} & \text{ if }h =4 \text{ and } c\ge 1,\\
        \emptyset &   \text{ otherwise.}\end{array}\right.
\]
\[
L_2(h,c,\turn, \nwstate)=\left\{\begin{array}{ll} 
\{A \} &\text{ if } h= 2, 4 \text{ and }c\ge 1,\\
\{B\} &\text{ if }h =6 \text{ and } c\ge 1,\\
        \emptyset &   \text{ otherwise.}\end{array}\right.
\]

In words, the decoy is set to host $4$.  The attacker misperceives the decoy host 4 as one critical
target, with the same label as host 2.

Similar to the previous experiment, we compare the set of states from
which P1 has a winning strategy to achieve his two objectives: safety
$\commspec$ and preferred $\commspec\land \hidspec$. The result is
shown in Table~\ref{tb:compare-complex-specs}. It turns out that with
deception, we see that the deceptive winning region for P1 given either
  objective  is larger than the
winning regions of games in which P2 has no misperception. At the beginning of
this game, the attacker is at host 0 with user access. Each host runs
all available services.  P1 can prevent P2 from achieving her co-safe
objective from the initial state with probability one. However, P1
cannot ensure to force P2 to running into the decoy  without
deception. When P2 misperceives the label of states, no matter P2 uses a sure-winning or an almost-sure winning randomized strategy, P1's deceptive winning regions for both safety and the preferred objectives include the initial state of the game. That is, P1 can prevent P2 from achieving her co-safe objective and force P2 to visit the decoy host $4$.
\begin{table}\centering
\caption{
                                        \label{tb:compare-complex-specs} Comparison P1's winning regions with/without deception.}
      \resizebox{0.75\textwidth}{!}{\begin{minipage}{\textwidth}
                                          \begin{tabular}{c | cc | cc|cc}
\toprule
                                     $\abs{V}$ &      \multicolumn{2}{c}{No Misperception} & \multicolumn{2}{c}{Greedy} &
                                 \multicolumn{2}{c}{Randomized} 
                                 \\
                                              \hline
                                       &   $\win_1 $ & $\win_1^\succ$    &   $\win_1$ &  $\win_1^\succ$  &    
                                     $\win_1$ &  $\win_1^\succ  $   \\
                                               28519 &  17013, win & 10782, lose  & 17170, win &  10964, win & 17013, win   & 10802, win \\
                                          \bottomrule
                                          \end{tabular}
                                        \end{minipage}}
                                    \end{table}
                                        
The code is implemented in python with a MacBook Air with 1.6 GHz dual-core 8th-generation Intel Core i5 Processor and 8GB  memory.  In the first experiment, the game arena (attack graph) is generated in $2.34$ sec and the hypergame transition system is generated in $8.16$ sec. The computation of winning regions and winning strategies for different P1 and P2 objectives took  $0.02$ to $0.04$ sec. In the second experiment, the   game arena (attack graph) is generated in $70.2$ sec and the hypergame transition system is generated in $854.83$ sec. The computation of winning regions and winning strategies for different P1 and P2 objectives took $3$ to $10$ sec.

\section{Conclusion}
\label{sec:conclude}
% We presented a framework of hypergame for synthesizing provably safe and correct security strategies from temporal logic objectives using active cyberdeception. 

% At this stage, the following work has been planned:
% \begin{itemize}
% \item Conduct experiments for larger synthetic network systems ($10$ to $20$ hosts) and more complex temporal logic objectives, and other defenses such as user access control, network reconfiguration. A real network system with concrete defenses and vulnerability definitions is desirable. 
% \item Theoretical analysis for simultanous attacks--this can be achieved by augmenting the states of hypergame transition systems with finite-memory encoding from which set of hosts attacks can be carried out. This will likely to create scalability issue given the attacker's action set grows combinatorially.
% \item Research about resource allocation: How to place the set of high-fidelity honeypots to ensure the existence of a winning strategy with cyberdeception for the defender given the initial state of the system?
% \item Theoretical analysis of permissive strategies in concurrent/simultanuous-move games and generalization to concurrent games.
% \item Incorporate learning/inference mechanisms for repeated attack games.
% \item Generalize to  games with partial observations.
% \item In active deception, the defender knows the identify of the user. It is also interesting to study identification with deception that resolves the uncertainty in the user's type--a legitimate user or an attacker.
% \end{itemize}

In this paper, we have developed the theory and solutions of $\omega$-regular hypergame and algorithms for synthesizing provably correct defense policies with active cyberdeception. Building on attack graph analysis in formal verification, we have shown that by modeling the active deception as an $\omega$-regular game in which the attacker misperceives the labeling function, the solution of such games can be used effectively to design deceptive security strategies under complex security requirements in temporal logic. We introduced a set-based strategy that approximates attacker's all possible rational decisions in her perceptual game. The defender's deceptive strategy against this set-based strategy, if exists, can ensure the security properties to be satisfied with probability one. We have experiments with synthetic network systems to verify the correctness of the policy and the advantage of deception. Our result is the first to integrate formal synthesis for $\omega$-regular games in designing secure-by-synthesis network systems using attack graph modeling.

The modeling and solution approach can be further generalized to consider  partial observations during interaction, concurrent interactions, and multiple target attacks from the attacker. The mechanism design problem in such a game is to be investigated for resource allocation in the network. Further, we will investigate methods to improve the scalability of the solution. One approach for tractable synthesis for such games is to use hierarchical aggregation in attack graphs \cite{noelManagingAttackGraph2004}. From formal synthesis, it is also possible to improve the scalability of algorithm with compositional synthesis \cite{Filiot2011,kulkarni2018compositional} and abstraction methods \cite{clarke2000counterexample} of $\omega$-regular games.

\bibliography{refs}{}

\begin{thebibliography}{10}

\bibitem{AitMaalemLahcen2018}
Rachid {Ait Maalem Lahcen}, Ram Mohapatra, and Manish Kumar.
\newblock {\em {Cybersecurity: A survey of vulnerability analysis and attack
  graphs}}, volume 253.
\newblock Springer Singapore, 2018.

\bibitem{al-shaerDynamicBayesianGames2019}
Ehab {Al-Shaer}, Jinpeng Wei, Kevin~W. Hamlen, and Cliff Wang.
\newblock Dynamic {{Bayesian Games}} for {{Adversarial}} and {{Defensive Cyber
  Deception}}.
\newblock In Ehab {Al-Shaer}, Jinpeng Wei, Kevin~W. Hamlen, and Cliff Wang,
  editors, {\em Autonomous {{Cyber Deception}}: {{Reasoning}}, {{Adaptive
  Planning}}, and {{Evaluation}} of {{HoneyThings}}}, pages 75--97. {Springer
  International Publishing}, {Cham}, 2019.

\bibitem{bennett1980hypergames}
Peter~G Bennett.
\newblock {Hypergames: developing a model of conflict}.
\newblock {\em Futures}, 12(6):489--507, 1980.

\bibitem{buchi1969solving}
J~Richard B{\"u}chi and Lawrence~H Landweber.
\newblock Solving sequential conditions by finite-state strategies.
\newblock {\em Transactions of the American Mathematical Society},
  138:295--311, 1969.

\bibitem{carrollGameTheoreticInvestigation2009}
Thomas~E. Carroll and Daniel Grosu.
\newblock A {{Game Theoretic Investigation}} of {{Deception}} in {{Network
  Security}}.
\newblock In {\em 2009 {{Proceedings}} of 18th {{International Conference}} on
  {{Computer Communications}} and {{Networks}}}, pages 1--6, August 2009.

\bibitem{chakraborty2018hybrid}
Tanmoy Chakraborty, Sushil Jajodia, Noseong Park, Andrea Pugliese, Edoardo
  Serra, and VS~Subrahmanian.
\newblock Hybrid adversarial defense: Merging honeypots and traditional
  security methods.
\newblock {\em Journal of Computer Security}, 26(5):615--645, 2018.

\bibitem{clarke2000counterexample}
Edmund Clarke, Orna Grumberg, Somesh Jha, Yuan Lu, and Helmut Veith.
\newblock Counterexample-guided abstraction refinement.
\newblock In {\em International Conference on Computer Aided Verification},
  pages 154--169. Springer, 2000.

\bibitem{Cohen2006}
Fred Cohen.
\newblock {The Use of Deception Techniques: Honeypots and Decoys}.
\newblock In {\em Handbook of Information Security 3.1}. 2006.

\bibitem{esparza2016ltl}
Javier Esparza, Jan K{\v{r}}et{\'\i}nsk{\`y}, and Salomon Sickert.
\newblock {From LTL to deterministic automata}.
\newblock {\em Formal Methods in System Design}, 49(3):219--271, 2016.

\bibitem{Filiot2011}
Emmanuel Filiot, Naiyong Jin, and Jean~Fran{\c{c}}ois Raskin.
\newblock {Antichains and compositional algorithms for LTL synthesis}.
\newblock {\em Formal Methods in System Design}, 39(3):261--296, 2011.

\bibitem{ijcai2019-50}
Karel Hor{\'{a}}k, Branislav Bo{\v{s}}ansk{\'{y}}, Christopher Kiekintveld, and
  Charles Kamhoua.
\newblock {Compact Representation of Value Function in Partially Observable
  Stochastic Games}.
\newblock In {\em Proceedings of the Twenty-Eighth International Joint
  Conference on Artificial Intelligence}, pages 350--356. International Joint
  Conferences on Artificial Intelligence Organization, 2019.

\bibitem{ingols2006practical}
Kyle Ingols, Richard Lippmann, and Keith Piwowarski.
\newblock {Practical attack graph generation for network defense}.
\newblock In {\em The 22nd Annual Computer Security Applications Conference},
  pages 121--130. IEEE, 2006.

\bibitem{Jajodia2016}
Sushil Jajodia, V.~S. Subrahmanian, Vipin Swarup, and Cliff Wang.
\newblock {\em {Cyber deception: Building the scientific foundation}}.
\newblock Springer, 2016.

\bibitem{Jha2002}
S.~Jha, O.~Sheyner, and J.~Wing.
\newblock {Two formal analyses of attack graphs}.
\newblock {\em Proceedings of the Computer Security Foundations Workshop},
  2002-Jan:49--63, 2002.

\bibitem{kulkarni2018compositional}
Abhishek~Ninad Kulkarni and Jie Fu.
\newblock {A Compositional Approach to Reactive Games under Temporal Logic
  Specifications}.
\newblock In {\em American Control Conference}, pages 2356--2362. IEEE, 2018.

\bibitem{kupferman2001model}
Orna Kupferman and Moshe~Y Vardi.
\newblock {Model checking of safety properties}.
\newblock {\em Formal Methods in System Design}, 19(3):291--314, 2001.

\bibitem{liu2016leveraging}
Jiaqiang Liu, Yong Li, Huandong Wang, Depeng Jin, Li~Su, Lieguang Zeng, and
  Thanos Vasilakos.
\newblock Leveraging software-defined networking for security policy
  enforcement.
\newblock {\em Information Sciences}, 327:288--299, 2016.

\bibitem{manna2012temporal}
Zohar Manna and Amir Pnueli.
\newblock {\em {The temporal logic of reactive and concurrent systems:
  Specification}}.
\newblock Springer Science {\&} Business Media, 2012.

\bibitem{ning2004techniques}
Peng Ning, Yun Cui, Douglas~S Reeves, and Dingbang Xu.
\newblock Techniques and tools for analyzing intrusion alerts.
\newblock {\em ACM Transactions on Information and System Security},
  7(2):274--318, 2004.

\bibitem{noelManagingAttackGraph2004}
Steven Noel.
\newblock Managing {{Attack Graph Complexity}} through {{Visual Hierarchical
  Aggregation}}.
\newblock In {\em In {{VizSEC}}/{{DMSEC}} '04: {{Proceedings}} of the 2004
  {{ACM}} Workshop on {{Visualization}} And}, pages 109--118. {ACM Press},
  2004.

\bibitem{pawlickGametheoreticTaxonomySurvey2019}
Jeffrey Pawlick, Edward Colbert, and Quanyan Zhu.
\newblock A {{Game}}-theoretic {{Taxonomy}} and {{Survey}} of {{Defensive
  Deception}} for {{Cybersecurity}} and {{Privacy}}.
\newblock {\em ACM Comput. Surv.}, 52(4):82:1--82:28, August 2019.

\bibitem{Pnueli1989}
A.~Pnueli and R.~Rosner.
\newblock On the synthesis of a reactive module.
\newblock In {\em Proceedings of the 16th ACM SIGPLAN-SIGACT Symposium on
  Principles of Programming Languages}, POPL '89, pages 179--190, New York, NY,
  USA, 1989. ACM.

\bibitem{Ramos2017}
Alex Ramos, Marcella Lazar, Raimir~Holanda Filho, and Joel~J.P.C. Rodrigues.
\newblock {Model-Based Quantitative Network Security Metrics: A Survey}.
\newblock {\em IEEE Communications Surveys and Tutorials}, 19(4):2704--2734,
  2017.

\bibitem{Sandholm99distributedrational}
Tuomas~W. Sandholm.
\newblock Distributed rational decision making, 1999.

\bibitem{Sheyner2002}
O.~Sheyner, J.~Haines, S.~Jha, R.~Lippmann, and J.M. Wing.
\newblock {Automated generation and analysis of attack graphs}.
\newblock In {\em Proceedings 2002 IEEE Symposium on Security and Privacy},
  pages 273--284. IEEE Computer Society, 2002.

\bibitem{Sheyner2004}
Oleg~Mikhail Sheyner.
\newblock {\em {Scenario graphs and attack graphs}}.
\newblock PhD thesis, Carnegie Mellon University, 2004.

\bibitem{underbrink2016effective}
AJ~Underbrink.
\newblock Effective cyber deception.
\newblock In {\em Cyber Deception}, pages 115--147. Springer, 2016.

\bibitem{Vane}
Russell Richardson~III Vane.
\newblock {\em {Using Hypergames to Select Plans in Competitive Environments}}.
\newblock PhD thesis, George Mason University, 2000.

\end{thebibliography}
\bibliographystyle{plain}

\appendix
\noindent \paragraph{Proof of Lemma~\ref{lma:determinism}}
By way of contradiction. If there exists another $\sigma'' \ne \sigma'$ such that  $\delta(q_2,\sigma'')=q_2''$ and $q_2''\ne q_2'$ and  $\mask(\sigma)=\mask(\sigma'')$. Given $\mask(\sigma')=\mask(\sigma'')=\mask(\sigma)$, P2 cannot distinguish $\sigma'$ from $\sigma''$--that is, $\sigma'=\sigma''$ in $\calA_2$. Thus,  $q_2'\ne q_2''$ contradicts the fact that $\calA_2$ is deterministic in the product automaton.
 
\noindent \paragraph{Algorithms for solving   co-safe/reachability and safety games}
Consider a two-player turn-based game arena $\game=(S= S_1\cup S_2,A = A_1\cup A_2, T)$ where for $i=1,2$, $S_i$ is a set of states where player $i$  takes an action, $A_i$ is a set of actions for player $i$, and $T:S\times A\rightarrow S$ is the transition function. 

We first define two functions:
\[
\pre_i^\exists (X) = \{ s\in S_i \mid \exists a \in A_i, \text{such that } T(s,a)\in X\};\]
which is a set of states from which player $i$ has an action to ensure reaching the set $X$ in one step.
\[
\pre_i^\forall (X) = \{ s\in S_i \mid \forall a \in A_i, \text{such that } T(s,a)\in X\};\]
which is a set of states from which all actions of player $i$ will lead to a state within $X$. 

Algorithm~\ref{alg:zielonka} is Zielonka's attractor algorithm, for solving player 1's winning region given a reachability objective. 
Algorithm~\ref{alg:safe} is the solution for turn-based games with safety objective for player 1. 
\begin{algorithm}
\caption{\label{alg:zielonka}Almost-Sure Winning for player 1 given the winning condition $(F,\cosafe)$ (reaching $F$).}
\begin{algorithmic}[1]
\item[\textbf{Inputs:}] $\left(\game=(X= X_1\cup X_2, A_1\cup A_2, T), F\subseteq S\right)$. 
    \State $Z_0 = F$
    \While{True} 
        \State $Z_{k+1} = Z_k \cup \pre^\exists_1(Z_k) \cup \pre^\forall_2(Z_k)$
        \If {$Z_{k+1} = Z_k$}
          \State $\win_1 \leftarrow Z_k$;
            \State End Loop
        \EndIf
        \State $k\leftarrow k+1$
    \EndWhile
      \For{$j = 1,\ldots, k$}
 \For{$x\in Z_j\cap X_1$}
   \State $\pi_1(x) = a \text{ if } T(x,a)\in Z_{j-1}$. \Comment{a sure-winning action from state $x$.}
   \EndFor
   \EndFor
    \State \Return $\win_1, \pi_1$. \Comment{Return the winning region and  strategy for player 1.}
\end{algorithmic}
\end{algorithm}
\begin{algorithm}
\caption{Almost-Sure Winning for player 1 given the winning condition $(B,\safe)$ (staying in $B$)  }
\label{alg:safe}
\begin{algorithmic}[1]
\item[\textbf{Inputs:}]  $\left(\game=(X= X_1\cup X_2, A_1\cup A_2, T), B \subseteq S\right)$. 
    \State  $Z_0= B$.
    \While{True}
    \State $Y = S\setminus Z_i$;  \Comment{$Y$ is a set of unsafe states.}
        \State $Z_{i+1} \leftarrow Z_i \setminus \left(Y\cup (\pre_{1}^\forall (Y) \cup \pre_2^\exists (Y))\right)$;
        \\
        \Comment{remove unsafe states from the set $Z_i$.}
        \If {$Z_{i+1} = Z_i$}
        \State $\win_1 \leftarrow Z_i$;
            \State End Loop
        \EndIf
        \State $i\leftarrow i+1$;
    \EndWhile
   \For{$x\in \win_1\cap X_1$}
   \State $\pi_1(x) =\{a\mid T(x,a)\in \win_1\}$;
   \EndFor
    \State \Return $\win_1, \pi_1$. \Comment{Return the winning region and  strategy for player 1.}
\end{algorithmic}
\end{algorithm}

Figure~\ref{fig:arena_simple} is the game arena for experiment 1. The set of states where the attacker makes a move is shown in grey box. The rest are the defender's states. Since the transition is deterministic, we omit the labels on the transitions.
\begin{figure*}
\includegraphics[width=\textwidth]{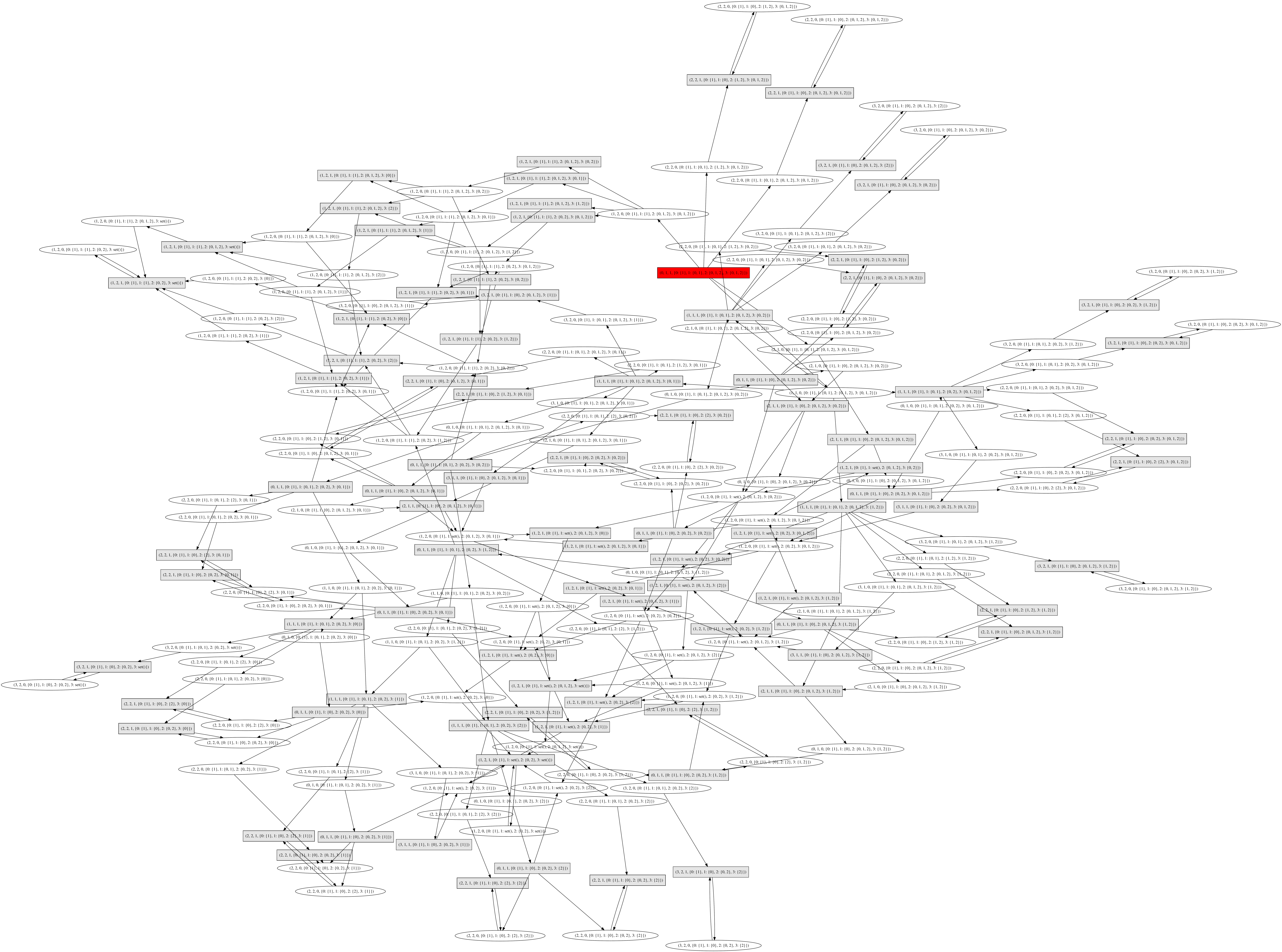}
\caption{The game arena of Experiment 1. \label{fig:arena_simple}}
\end{figure*}

Figure~\ref{fig:game_win} is the hypergame transition system given P1's sure winning strategy and P2's perceived winning strategy in experiment 1. 
%\begin{landscape}

\begin{figure*}[ht]
\centering
\includegraphics[angle=-90,scale=0.08]{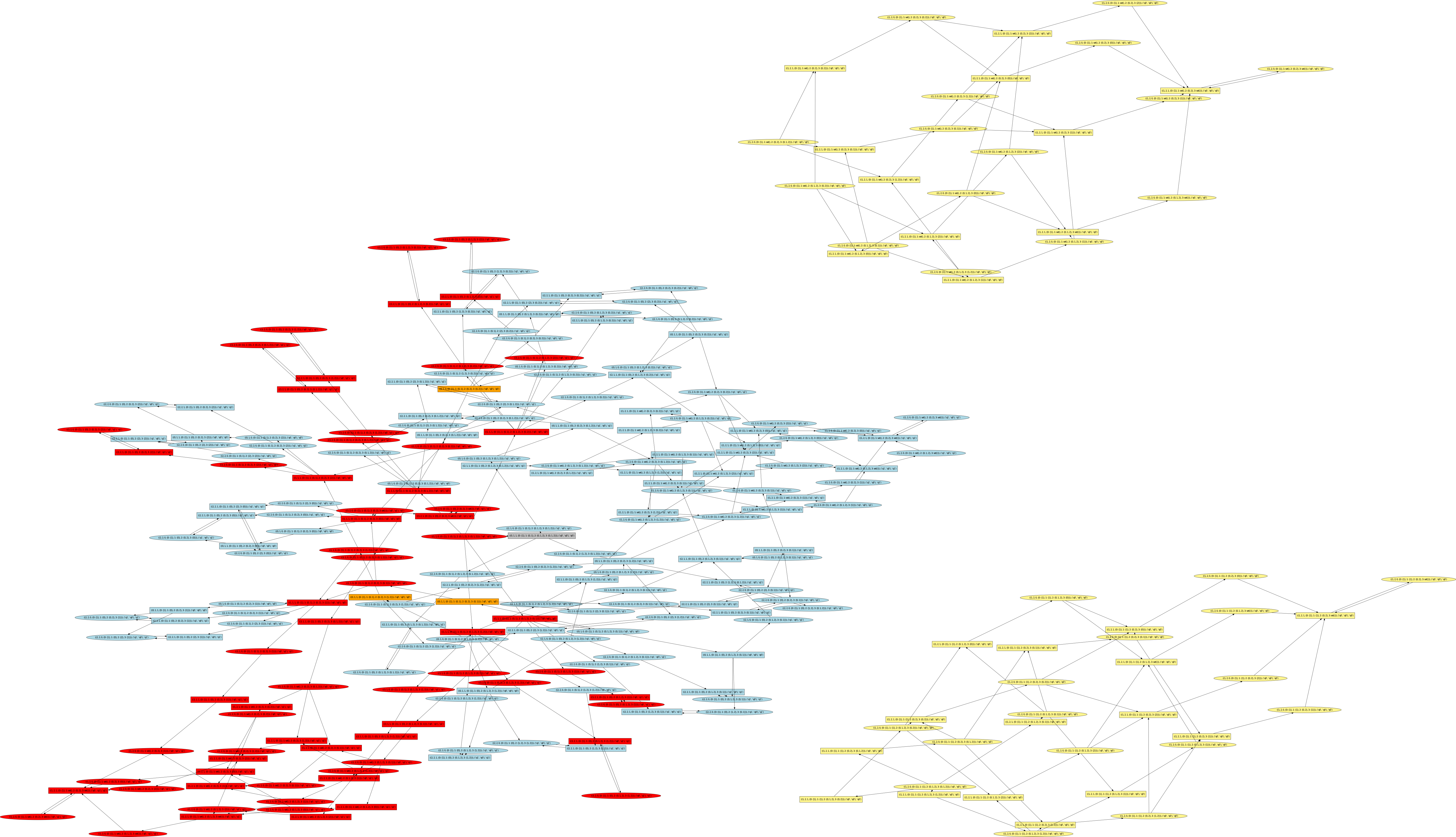}
\caption{The winning regions of two players in Experiment 1. \label{fig:game_win}}
\end{figure*}
%\end{landscape}

\end{document}